\documentclass[12pt]{article}
\usepackage{amsfonts,amsmath,graphicx,amssymb,geometry,arydshln,textcomp,amsthm,hyperref,framed,color,appendix,lscape}
\usepackage[round, authoryear]{natbib}
\usepackage[nodisplayskipstretch]{setspace}
\usepackage{ragged2e}
\usepackage{xcolor,xspace,colortbl,rotating}
\usepackage{tabulary}
\usepackage{booktabs}
\usepackage{multirow}
\usepackage{bm}
\usepackage[ruled,linesnumbered]{algorithm2e} 
\usepackage{cleveref}
\usepackage{threeparttable}
\usepackage[subrefformat=parens,labelformat=parens]{subfig}
\usepackage{etoc}
\usepackage{enumerate}
\usepackage{ltablex,caption} 
\DeclareGraphicsExtensions{.pdf,.svg,.eps,.ps,.png,.jpg,.jpeg}
\theoremstyle{definition}
\newtheorem{theorem}{Theorem}

\newtheorem{assumption}{Assumption}

\newtheorem {corollary}{Corollary}

\newtheorem {lemma}{Lemma}

\newtheorem {remark}{Remark}

\crefname{subequation}{subequation}{subequations}
\Crefname{subequation}{Assumption}{assumptions}
\crefname{assumption}{Assumption}{Assumptions}
\hypersetup{colorlinks,linkcolor={blue},citecolor={blue},urlcolor={blue}}

\DeclareMathOperator{\tr}{tr}

\makeatletter
\newcommand\subsubsubsection{\@startsection{paragraph}{4}{\z@}{-2.5ex\@plus -1ex \@minus -.25ex}{1.25ex \@plus .25ex}{\normalfont\normalsize\bfseries}}
\newcommand\subsubsubsubsection{\@startsection{subparagraph}{5}{\z@}{-2.5ex\@plus -1ex \@minus -.25ex}{1.25ex \@plus .25ex}{\normalfont\normalsize\bfseries}}
\makeatother

\usepackage{etoolbox}
\BeforeBeginEnvironment{equation*}{\begin{singlespace}\vspace*{-\baselineskip}}
\AfterEndEnvironment{equation*}{\end{singlespace}\noindent\ignorespaces}
\BeforeBeginEnvironment{equation}{\begin{singlespace}\vspace*{-\baselineskip}}
\AfterEndEnvironment{equation}{\end{singlespace}\noindent\ignorespaces}
\BeforeBeginEnvironment{align*}{\begin{singlespace}\vspace*{-\baselineskip}}
\AfterEndEnvironment{align*}{\end{singlespace}\noindent\ignorespaces}
\BeforeBeginEnvironment{align}{\begin{singlespace}\vspace*{-\baselineskip}}
\AfterEndEnvironment{align}{\end{singlespace}\noindent\ignorespaces}

\begin{document} 

\author{Xiao Huang\thanks{Department of Economics, Finance, and Quantitative Analysis, Coles College of Business, Kennesaw State University, GA 30144, USA. Email: xhuang3@kennesaw.edu.}}
\title{\Large Boosted \textit{p}-Values for High-Dimensional Vector Autoregression}
\date{ \today}
\maketitle

\doublespacing

\begin{abstract}
	Assessing the statistical significance of parameter estimates is an important step in high-dimensional vector autoregression modeling. Using the least-squares boosting method, we compute the \textit{p}-value for each selected parameter at every boosting step in a linear model. The \textit{p}-values are asymptotically valid and also adapt to the iterative nature of the boosting procedure. Our simulation experiment shows that the \textit{p}-values can keep false positive rate under control in high-dimensional vector autoregressions. In an application with more than $100$ macroeconomic time series, we further show that the \textit{p}-values can not only select a sparser model with good prediction performance but also help control model stability. A companion R package \texttt{boostvar} is developed.

\end{abstract}

\bigskip


\bigskip


\newpage  

\normalsize

\doublespacing

\section{Introduction} \label{sec:intro}

Since its introduction as a macroeconometric framework in \cite{sims1980var}, vector autoregression (VAR) has become one of the most valuable econometric tools for describing the joint dynamics of multiple time series. Its application can also be found in many other disciplines such as biology and neuronscience. With increasingly available data, there have been many developments in estimating high-dimensional VARs. Ideas of using the lasso in \cite{tibshirani1996lasso} or imposing other penalties are flourishing. The use of the lasso or its variants can be found in \cite{fujita2007,lozano2009grplasso,basu2015lasso,medeiros2016lasso,nicholson2020grplasso,wong2020lasso}, among others. The Dantzig selector in \cite{han2015directVAR} and the tensor method in \cite{wang2021tensor} also directly address the estimation problem in high-dimensional VARs. Methods such as those in \cite{chen2012srrr,uematsu2019sofar} are designed for multivariate linear regression, and they can also be used for estimation. A data analyst now has many methods for estimation in high-dimensional VARs.

One unanswered question in this active area of research is how to assign \textit{p}-values to parameter estimates in a high-dimensional VAR. In a standard linear regression, we provide both parameter estimates and their \textit{p}-values, and the same should be expected in high-dimensional VARs. Without a \textit{p}-value, it is not possible to assess the statistical uncertainty associated with an estimate. If the research purpose is model interpretation, obtaining a valid \textit{p}-value is perhaps as important as having a parameter estimate. If the goal is prediction, the model is typically selected via cross-validation (c.v.). However, cross-validation itself guarantees no statistical significance, and knowing the \textit{p}-values can help improve the prediction for certain data. Given the abundant evidence in the literature on cross-section regression that tools like the lasso can have high false positive rate, it is reasonable to conjecture that many of the non-zero parameter estimates produced by the lasso (and some other methods such as boosting) can be statistically insignificant. A \textit{p}-value can certainly help reveal better the underlying structure in the data.

Since each estimation method generates its own sparsity, it is clear that we need different ways of computing the standard error (s.e.) to reflect the unique uncertainty embedded in each method. We choose to work with the least-squares boosting (LS-Boost) method in this paper. LS-Boost is introduced in \cite{friedman2001gbm} as an example of gradient boosting algorithm. It can easily handle a regression model with high dimension. Its application in high-dimensional VAR estimation is discussed in \cite{lutz2006boosting}. The focus of this paper is not to develop another algorithm to generate a new sparsity in high-dimensional VARs. Instead, we study the following question: If LS-Boost is used in a high-dimensional VAR for estimation, can we derive a standard error and compute the \textit{p}-value for each selected parameter at every boosting step? We give an affirmative answer to this question.

When the number of boosting step is large, the LS-Boost estimator converges to the LS estimator. We provide a simulation example to show our \textit{p}-value can converge to the \textit{p}-value of the LS estimator in a bivariate VAR. When the VAR dimension is large, there is no unique LS solution to a VAR; the LS-Boost estimator converges to one of the LS solutions as boosting step increases. For LS-Boost, our method can compute the \textit{p}-value for every estimate at a given boosting step. Just like the lasso and many other penalized estimator, the LS-Boost estimator is biased at a given boosting step. In this case, the \textit{t}-test uses a biased parameter estimator to test the null hypothesis of whether the cumulative incremental changes in a parameter is statistically different from $0$ at every boosting step (see more detailed discussions in \Cref{sec:discussions}). Let $\beta_j$ be a scalar parameter for the variable $X_j$. The classical fixed hypothesis is $H_0{:}\;\beta_j = 0$. One is supposed to use an unbiased estimator $\hat{\beta}_j$ for hypothesis testing. In this regard, our null hypothesis is not equivalent to $H_0{:}\;\beta_j = 0$. This, however, should not be the ground on which to discount the usefulness of our \textit{p}-value. Consider a case of a cross-section regression with $1,000$ variables and $100$ observations. To test $H_0{:}\;\beta_j = 0$, we need a reliable $\hat{\beta}_j$ and its s.e. But to obtain a $\hat{\beta}_j$ in a high-dimensional regression, we will inevitably resort to some adaptive procedures such as the lasso. As soon as one uses the lasso for even an initial estimate, the lasso selection and estimation risk starts to contribute the overall estimation risk of $\hat{\beta}_j$, and these additional risks must be accounted for when testing the null, an intrinsically hard problem. In addition, no matter which penalized method is used in high-dimensional model, one will likely end up with a biased estimator for $\beta_j$ and use it in fixed hypothesis testing. As a result, hypothesis testing in high dimensional models is challenging. Various approaches have been proposed for fixed hypotheses testing in high-dimensional cross-section regressions, including \cite{wasserman2009hdselection,meinshausen2009pval,zhangzhang2014JRSSBlowdimension,vandeGeer2014AOSoptimalinterval}, among others. These approaches typically make some assumptions on the sparsity and the variables of the model. These papers also have one thing in common: they all target the unknown \textit{population} parameter $\beta_j$ and try to test the hypothesis $H_0{:}\;\beta_j = 0$.


We take a different approach in this paper. In high dimensional linear models, since it is practically impossible to obtain an unbiased estimator for $\beta_{j}$, can we estimate $\beta_j$ partially, in an incremental way, and perform a \textit{t}-test on these partial estimates? The new perspective this paper explores is we can use a sequence of biased LS-Boost estimates to continuously check the statistical significance of a coefficient up to each boosting step. Put differently, if the numerator of a \textit{t}-statistic is biased, as long as we can obtain a valid s.e. for the biased estimator and correctly quantify the statistical uncertainty of the numerator, the \textit{t}-statistic will still have an asymptotic normal distribution. This argument will be made more precise in the discussion in \Cref{sec:discussions}.


A related line of research is the recent development in post-selection inference for the lasso, where works in \cite{lockhart2014lassopval,lee2016lassopval,tibshiranietal2016lassopval} focus on how to assign a valid \textit{p}-value to the lasso estimator by adaptively incorporating the uncertainty in variable selection. Our work differs from theirs in several aspects: (i) this paper studies the issue of assigning \textit{p}-values in a VAR for LS-Boost while their work is focused on the lasso (and several other sequential procedures) in a cross-section regression; (ii) more importantly, \cite{lee2016lassopval,tibshiranietal2016lassopval} assume that data are normal and use polyhedral conditioning sets to characterize the lasso selection event and to address the post-selection inference in the lasso; whereas we make no assumption of data normality and derive the estimator's variance directly from its closed-form solution; (iii) we also provide a tractable asymptotic distribution for each parameter estimate at every boosting step. Our method of computing the \textit{p}-values in VAR nests cross-section regression as a special case. In this special case, we use three common data sets to demonstrate the computation of \textit{p}-values in a cross-section regression. The discussion can be found in the online supplement.

This paper makes the following contributions to the literature on high-dimensional VARs. First, we provide an asymptotic distribution result for every selected parameter at each boosting step for a stationary VAR. We do not make any sparsity assumption on the parameter; nor do we impose any assumption on the relative growth rate between sample size and the dimension of the model. The asymptotic variance allows us to construct a \textit{t}-test easily. Second, we extend the computational bound results for boosting in \cite{freundetal2017boosting} to VAR. When combined with the asymptotic result, this allows us to characterize the behavior of LS-Boost estimator when both the boosting step and sample size go to infinity. Third, our simulation result shows that using a \textit{p}-value can help reduce the false positive rate (FPR) and improve the F score of a model. Our application using the macroeconomic data in \cite{mccracken2016data} further demonstrates that a \textit{p}-value-adjusted model can give good prediction performance and good model stability. With the help of \textit{p}-value, we can remove more than $95\%$ of the nonzero LS-Boost estimates while still improving the prediction in this application.  An R package \texttt{boostvar} that implements our method can be found at \url{https://github.com/xhuang20/boostvar}.

The rest of the paper is organized as follows. Section 2 discusses two LS-Boost algorithms in VAR. Section 3 derives the standard error and asymptotic distribution for the LS-Boost estimator. Section 4 presents the computation bounds for boosting in a VAR model. Simulation results are discussed in Section 5, and Section 6 includes an application. Section 7 concludes. The online supplement contains all proofs, additional discussions and tables.

\section{Two LS-Boost algorithms in VAR} \label{sec:LS-Boost}

We begin by introducing the standard VAR model and some notations. Consider a $d\times1$ vector time series $y_t$ with $t=1,\cdots,T$. A $p$th-order VAR (VAR($p$)) is defined as
\begin{equation} \label{eq:varp}
	y_t = \phi_1 y_{t-1} + \cdots + \phi_p y_{t-p} + u_t,
\end{equation}
where $\phi_1,\cdots, \phi_p$ are $d \times d$ matrices of coefficients and $u_t$ is a $d \times 1$ vector of error terms. We omit the intercept in formulating the model by assuming data are demeaned. The total number of parameters is $p \times d^2$, which can easily exceed the number of observations $T$ and make the LS method infeasible.

It is often convenient to rewrite \cref{eq:varp} in a matrix format. Let $'$ denote matrix/vector transpose. Define

\begin{equation*} \label{eq:some matrices}
	\underset{T \times d}{\mathbf{Y}} = \begin{bmatrix}
	y_1'\\
	\vdots\\
	y_T'
	\end{bmatrix},
	\underset{T \times d}{\mathbf{Y}_{-1}} = \begin{bmatrix}
	y_0'\\
	\vdots\\
	y_{T-1}'
	\end{bmatrix},
	\cdots,
	\underset{T \times d}{\mathbf{Y}_{-p}} = \begin{bmatrix}
	y_{1-p}'\\
	\vdots\\
	y_{T-p}'
	\end{bmatrix}, 
	\underset{T \times d}{\mathbf{u}} = \begin{bmatrix}
	u_1'\\
	\vdots\\
	u_T'
	\end{bmatrix}, \text{ and }
	\underset{pd \times d}{\bm{\phi}} = \begin{bmatrix}
	\phi_1'\\
	\vdots \\
	\phi_p'
	\end{bmatrix}.
\end{equation*}
In matrix form, \cref{eq:varp} can be written as
\begin{equation} \label{eq:varp_transpose}
	\mathbf{Y} =  \mathbf{Y}_{-1} \phi_1' + \cdots + \mathbf{Y}_{-p} \phi_p' + \mathbf{u}.
\end{equation}
The multivariate regression format is given by
\begin{equation} \label{eq:varp_matrix}
	\mathbf{Y} = \mathbf{X} \bm{\phi} + \mathbf{u},
\end{equation}
where 
\begin{equation} \label{eq:X}
	\mathbf{X}=\left[\mathbf{Y}_{-1},\cdots,\mathbf{Y}_{-p}\right] \text{ and } \bm{\phi} = [\phi_1, \cdots, \phi_p]'.
\end{equation} 

The basic idea of the LS-Boost is to recursively select a column of the covariates in $\mathbf{X}$ that gives the best LS fit of $\mathbf{Y}$ in the first boosting step or the current residuals in later boosting steps.
There is no unique way to implement the LS-Boost algorithm in a VAR. At each boosting step, one can select a column in $\mathbf{X}$ that gives the best fit for a selected column in $\mathbf{Y}$ or the residual matrix and update a single element in $\bm{\phi}$, which is the componentwise linear least squares procedure described in \cite{lutz2006boosting}. Alternatively, one can use the ``row-boosting" procedure in \cite{lutz2006boosting} to select a column in $\mathbf{X}$ to gives the best fit for the entire $\mathbf{Y}$ or the residual matrix and update a row of entries in $\bm{\phi}$. In the following, we describe a slightly more general procedure that nests the ``row-boosting" as a special case.

\subsection{A Group (by variable) LS-Boost Algorithm} \label{sec:group boosting}

To motivate the group LS-boost procedure, let us consider a VAR(2) model with three variables: unemployment rate (UNR), consumer price index (CPI), and inflation rate (INF). \Cref{eq:varp_matrix} becomes

\begin{equation}
	[\underbrace{\text{UNR}, \text{CPI}, \text{INF}}_{T \times 3}] = [\underbrace{\text{UNR}_{-1}, \text{CPI}_{-1}, \text{INF}_{-1},\text{UNR}_{-2}, \text{CPI}_{-2}, \text{INF}_{-2}}_{T \times 6}] \bm{\phi} + \mathbf{u}.
\end{equation}

Our first algorithm works by selecting a variable with all its lags that gives the best fit at each boosting step. Hence, if the variable CPI is found to give the best fit for $\mathbf{Y}$, the matrix $[\text{CPI}_{-1},\text{CPI}_{-2}]$ will be the selected variable matrix for that step. Since we select each variable with all its lags as a group, it will be helpful to rearrange columns in \cref{eq:X} so that they are grouped by variables
\begin{equation} \label{eq:Xg}
	\mathbf{X}_g = [\mathbf{X}_{(1)},\cdots,\mathbf{X}_{(j)},\cdots,\mathbf{X}_{(d)}]_{T \times pd},
\end{equation}
where $X_{(j)}$ is a $T \times p$ matrix that collects $p$ columns of lags in variable $j$ such as $[\text{CPI}_{-1},\text{CPI}_{-2}]$. And the VAR($p$) process in \cref{eq:varp} can be written as
\begin{align}
	\mathbf{Y} &=  \mathbf{X}_{(1)} \phi_{(1)}' + \cdots + \mathbf{X}_{(j)} \phi_{(j)}' + \cdots + \mathbf{X}_{(d)} \phi_{(d)}' + \mathbf{u}  \nonumber \\
	&= \mathbf{X}_g \bm{\phi}_g + \mathbf{u},  \label{eq:varp_grped}
\end{align}
where $\bm{\phi}_g = [\phi_{(1)},\cdots,\phi_{(j)},\cdots, \phi_{(d)}]'$. $\phi_{(j)}'$ is a $p \times d$ coefficient matrix, each row of which comes from the $j$th row of $\phi_1',\cdots,\phi_p'$.

Let $k$ denote the iteration step and $k=0,1,\cdots,k_{\text{stop}}$, and $k_{\text{stop}}$ is a prespecified stopping number in boosting iteration. Let the superscript $(k)$ denote a quantity associated with step $k$ so that, for example, $\hat{\mathbf{R}}^{(k)}$ is the $T \times d$ residual matrix generated at step $k$. Because we select all lags of a variable at each step, the LS regression at step $k$ will take the following form for variable $j$
\begin{equation} \label{eq:boost reg}
	\underset{T \times d}{\vphantom{\phi_{(j)}'}\hat{\mathbf{R}}^{(k-1)}} = \underset{T \times p}{\vphantom{\phi_{(j)}'}\mathbf{X}_{(j)}} \underset{p \times d}{\phi_{(j)}'} + \text{error term},
\end{equation}
where $ \hat{\mathbf{R}}^{(k-1)} = \hat{\mathbf{R}}^{(0)} = \mathbf{Y} $ if $ k=1 $.

If we define the $p \times d$ matrix
\begin{equation} \label{eq:A mat}
	\mathbf{A}_j = (\mathbf{X}_{(j)}'\mathbf{X}_{(j)})^{-1}\mathbf{X}_{(j)}',
\end{equation}
the estimate in \cref{eq:boost reg} with a learning rate $\nu$, can be updated as
\begin{equation} \label{eq:phi update}
	\hat{\phi}_{(j)}^{(k)\prime} = \hat{\phi}_{(j)}^{(k-1)\prime} + \nu \mathbf{A}_j \hat{\mathbf{R}}^{(k-1)}.
\end{equation}
Initialize $\hat{\phi}_{(1)}^{(0)} = \cdots = \hat{\phi}_{(d)}^{(0)} = \mathbf{0}_{d \times p} $. Let $\lVert \cdot \rVert_2$ denote the matrix and vector Euclidean norm. For each boosting iteration $k \ge 1$, the algorithm becomes
\begin{enumerate}[]
	\item[] \textit{Step 1}. Select the variable $\mathbf{X}_{(j_k)}$ with 
	\begin{equation} \label{eq:group boost obj}
	j_k \in \operatorname*{argmin}_{1 \leq j \leq d} \frac{1}{Td} \lVert \hat{\mathbf{R}}^{(k-1)} - \mathbf{X}_{(j)} \beta_{(j)}^{(k)\prime} \rVert_{2}^{2} \text{ with } \hat{\beta}_{(j)}^{(k)\prime} =\mathbf{A}_j \hat{\mathbf{R}}^{(k-1)}.
	\end{equation}
	\item[] \textit{Step 2}. Update $\hat{\phi}_{(j)}^{(k)}$ and $\hat{\mathbf{R}}^{(k)}$.
	\begin{align}
	\hat{\phi}_{(j)}^{(k)\prime} &= \begin{cases}
	\hat{\phi}_{(j_k)}^{(k-1)\prime}  + \nu \hat{\beta}_{(j_k)}^{(k)\prime} & \text{if } j = j_k, \\
	\hat{\phi}_{(j)}^{(k-1)\prime} & \text{if } j \neq j_k,
	\end{cases} \label{eq:group phi update} \\
	\hat{\mathbf{R}}^{(k)} &= \hat{\mathbf{R}}^{(k-1)} - \nu \mathbf{X}_{(j_k)} \hat{\beta}_{(j_k)}^{(k)\prime}. \label{eq:group R update}
	\end{align}
	Update several other quantities.
	\begin{align}
		\mathbf{A}_j^{(k)} &=\begin{cases}
		\mathbf{A}_{j_k}& \text{if } j = j_k,\\
		0& \text{if } j \ne j_k,
		\end{cases} \label{eq:group Ajk}\\
		\mathbf{H}^{(k)} &= \mathbf{X}_{(j_k)}(\mathbf{X}_{(j_k)}'\mathbf{X}_{(j_k)})^{-1}\mathbf{X}_{(j_k)}',\label{eq:group Hk}\\
		\mathbf{B}^{(k)} &= \mathbf{I}_T - (\mathbf{I}_T-\nu \mathbf{H}^{(k)})\cdots
		(\mathbf{I}_T-\nu \mathbf{H}^{(1)}), \label{eq:group Bk}\\
		\hat{\mathbf{F}}^{(k)} &= \hat{\mathbf{F}}^{(1)} + \cdots + \hat{\mathbf{F}}^{(k-1)} + \nu \mathbf{H}^{(k)} \mathbf{R}^{(k-1)} = \mathbf{B}^{(k)} \mathbf{Y},\label{eq:group Fk}\\
		\hat{\mathbf{u}}^{(k)} &= \mathbf{Y} - \hat{\mathbf{F}}^{(k)}= (\mathbf{I}_T - \mathbf{B}^{(k)}) \mathbf{Y} \label{eq:group uk}.
	\end{align}
	\item[] \textit{Step 3}. Repeat Steps 1 and 2 many times until step $k_{\text{stop}}$.
\end{enumerate}
We call this algorithm LS-Boost1. The matrix $\mathbf{B}^{(k)}$ in \cref{eq:group Bk} is the common hat matrix, and the sum of its diagonal elements gives the degree of freedom up to boosting step $k$.

\subsection{A sparser LS-Boost Algorithm} \label{sec:second alg}

An approach with a sparser solution is to, instead of using all lags of variable $\mathbf{X}_{(j)}$ in \cref{eq:boost reg} at each boosting step $k$, select only a single column of $\mathbf{X}_{(j)}$ to fit $\hat{\mathbf{R}}^{k-1}$. There is no guarantee that a recent lag will enter the model earlier than a distant lag. For example, the column $\text{CPI}_{-2}$ can enter the model earlier than $\text{CPI}_{-1}$ does. Let $\mathbf{X}_{(j)s}$ be the $s$th lag of the $j$th variable (the $s$th column in $\mathbf{X}_{(j)}$) and $\phi_{(j)s}'$ be the $s$th row of $\phi_{(j)}'$, the regression equation for the ``row-boosting" algorithm can be written as 
\begin{equation} \label{eq:row boost reg}
	\underset{T \times d}{\vphantom{\phi_{(j)}'}\hat{\mathbf{R}}^{(k-1)}} = \underset{T \times 1}{\vphantom{\phi_{(j)}'}\mathbf{X}_{(j)s}} \underset{1 \times d}{\phi_{(j)s}'} + \text{error term},
\end{equation}
where $\phi_{(j)s}'$ is a $1 \times d$ row vector, which is the $s$th column of $\phi_{(s)}$. Define
\begin{equation} \label{eq:A mat js}
	\underset{1 \times T}{\vphantom{\mathbf{X}_{(j)s}'}\mathbf{A}_{js}} = (\mathbf{X}_{(j)s}'\mathbf{X}_{(j)s})^{-1}\mathbf{X}_{(j)s}'.
\end{equation}
The update equation for the coefficient estimate at step $k$ is given by
\begin{equation} \label{eq:phi update alg2}
\hat{\phi}_{(j)s}^{(k)\prime} = \hat{\phi}_{(j)s}^{(k-1)\prime} + \nu \mathbf{A}_{js} \hat{\mathbf{R}}^{(k-1)}.
\end{equation}
Initialize $\hat{\phi}_{(j)s}^{(0)} = \mathbf{0}_{d \times 1} \forall j=1,\cdots,d$ and $s = 1,\cdots,p$.  For each boosting step $k \ge 1$, the algorithm becomes
\begin{enumerate}[]
	\item[] \textit{Step 1}. Select the variable $\mathbf{X}_{(j_k)s_k}$ with 
	\begin{equation} \label{eq:group boost obj js}
	j_k,s_k \in \operatorname*{argmin}_{\substack{1 \leq j \leq d \\ 1 \leq s \leq p}} \frac{1}{Td} \lVert \hat{\mathbf{R}}^{(k-1)} - \mathbf{X}_{(j)s} \beta_{(j)s}^{(k)\prime} \rVert_{2}^{2} \text{ with } \hat{\beta}_{(j)s}^{(k)\prime} =\mathbf{A}_{js} \hat{\mathbf{R}}^{(k-1)}.
	\end{equation}
	\item[] \textit{Step 2}. Update $\hat{\phi}_{(j)s}^{(k)}$ and $\mathbf{R}^{(k)}$.
	\begin{align}
	\hat{\phi}_{(j)s}^{(k)\prime} &= \begin{cases}
	\hat{\phi}_{(j_k)s_k}^{(k-1)\prime}  + \nu \hat{\beta}_{(j_k)s_k}^{(k)\prime} & \text{if } j = j_k \text{ and } s = s_k, \\
	\hat{\phi}_{(j)s}^{(k-1)\prime} & \text{otherwise},
	\end{cases} \label{eq:group phi update js} \\
	\hat{\mathbf{R}}^{(k)} &= \hat{\mathbf{R}}^{(k-1)} - \nu \mathbf{X}_{(j_k)s_k} \hat{\beta}_{(j_k)s_k}^{(k)\prime}. \label{eq:group R update js}
	\end{align}
	Update several other quantities.
	\begin{align}
	\mathbf{A}_{js}^{(k)} &=\begin{cases}
	\mathbf{A}_{j_ks_k}& \text{if } j = j_k \text{ and } s = s_k\\
	0& \text{otherwise},
	\end{cases} \label{eq:group Ajks}\\
	\mathbf{H}^{(k)} &= \mathbf{X}_{(j_k)s_k}(\mathbf{X}_{(j_k)s_k}'\mathbf{X}_{(j_k)s_k})^{-1}\mathbf{X}_{(j_k)s_k}',\label{eq:group Hks}\\
	\mathbf{B}^{(k)} &= \mathbf{I}_T - (\mathbf{I}_T-\nu \mathbf{H}^{(k)})\cdots
	(\mathbf{I}_T-\nu \mathbf{H}^{(1)}), \label{eq:group Bks}\\
	\hat{\mathbf{F}}^{(k)} &= \hat{\mathbf{F}}^{(1)} + \cdots + \hat{\mathbf{F}}^{(k-1)} + \nu \mathbf{H}^{(k)} \mathbf{R}^{(k-1)} = \mathbf{B}^{(k)} \mathbf{Y},\label{eq:group Fks}\\
	\hat{\mathbf{u}}^{(k)} &= \mathbf{Y} - \hat{\mathbf{F}}^{(k)}= (\mathbf{I}_T - \mathbf{B}^{(k)}) \mathbf{Y} \label{eq:group uks}.
	\end{align}
	\item[] \textit{Step 3}. Repeat Steps 1 and 2 many times until step $k_{\text{stop}}$.
\end{enumerate}
We call this algorithm LS-Boost2. Note that, although \cref{eq:group Bks,eq:group Fks,eq:group uks} may look the same as \cref{eq:group Bk,eq:group Fk,eq:group uk}, their values are different since \cref{eq:group Hks} is based on $\mathbf{X}_{(j_k)s_k}$.


A few remarks are in order.
\begin{remark}
	Since the dimension of $\phi_{(j)}$ in \cref{eq:boost reg} is $d \times p$, the total number of nonzero parameters will increase by $dp$ if a new variable is added. When the dimension of $d$ is large, $dp$ can be large. However, since we have $d$ equations, on average, each equation will add $p$ more parameters, and $p$ is usually a small number such as $2$ or $3$. Thus the model size for each of the $d$ equation will not increase significantly when using LS-Boost1.
\end{remark}

\begin{remark}
	LS-Boost2 is more flexible than LS-Boost1 and will generate a sparser model in general.  LS-Boost1 preserves the time series structure throughout variable selection and estimation, and it can work better for certain data.
\end{remark}

\begin{remark}
	The most flexible boosting algorithm is to update the parameters one at a time. While the previous two algorithms will add either $dp$ or $d$ parameters for a selected variable, we can choose to update a single parameter at each boosting iteration step, which will require us to take many sweeps in the data. This approach, called the componentwise procedure in \cite{lutz2006boosting}, completely ignores the time series structure of the data. Its implementation will be similar to that of LS-Boost in a cross-section regression. We focus on the use of LS-Boost1 and LS-Boost2 in this paper.
\end{remark}

\section{The asymptotic standard error and the \textit{p}-value}
We discuss the asymptotic result for the two boosting methods introduced in \Cref{sec:LS-Boost}.
\subsection{Asymptotic results for LS-Boost1 when $k$ is fixed} \label{sec:asym_lsboost1}
We first discuss the computation of the standard error for LS-Boost1 at each boosting step. Using the definition of $\hat{\beta}_{(j)}^{(k)\prime}$ in \cref{eq:group boost obj}, rewrite \cref{eq:phi update} in a recursive form to have
\begin{align} \label{eq:phi jk1}
	\hat{\phi}_{(j)}^{(k)\prime} &= \hat{\phi}_{(j)}^{(k-1)\prime} + \nu \hat{\beta}_{(j)}^{(k)\prime} = \hat{\phi}_{(j)}^{(k-2)\prime} + \nu \hat{\beta}_{(j)}^{(k-1)\prime} + \nu \hat{\beta}_{(j)}^{(k)\prime} \nonumber = \cdots \nonumber \\
	&= \hat{\phi}_{(j)}^{(0)\prime} + (\nu \hat{\beta}_{(j)}^{(1)\prime} + \cdots + \nu \hat{\beta}_{(j)}^{(k)\prime}) \nonumber \\
	&= \nu \hat{\beta}_{(j)}^{(1)\prime} + \cdots + \nu \hat{\beta}_{(j)}^{(k)\prime} \qquad\text{ since } \hat{\phi}_{(j)}^{(0)\prime}=0.
\end{align}
Following the definitions in \cref{eq:group Ajk,eq:group Hk,eq:group Bk,eq:group Fk,eq:group uk}, \cref{eq:phi jk1} can be further written as
\begin{align} \label{eq:phi jk2}
	\hat{\phi}_{(j)}^{(k)\prime} &=\nu \mathbf{A}_j^{(1)} \mathbf{Y} + \nu \mathbf{A}_j^{(2)}(\mathbf{I}_T - \nu \mathbf{H}^{(1)}) \mathbf{Y}+\nu \mathbf{A}_j^{(3)}(\mathbf{I}_T - \nu \mathbf{H}^{(2)})(\mathbf{I}_T - \nu \mathbf{H}^{(1)}) \mathbf{Y} \nonumber \\
	&\quad + \cdots + \nu \mathbf{A}_j^{(k)}(\mathbf{I}_T - \nu \mathbf{H}^{(k-1)})\cdots(\mathbf{I}_T - \nu \mathbf{H}^{(1)}) \mathbf{Y} \nonumber \\
	&=\tilde{\mathbf{A}}_{j}^{(k)} \mathbf{Y}, 
\end{align}
where we define
\begin{equation} \label{eq: A tilde jk}
	\tilde{\mathbf{A}}_{j}^{(k)} = \left[\nu \mathbf{A}_j^{(1)} + \nu \mathbf{A}_j^{(2)}(\mathbf{I}_T - \nu \mathbf{H}^{(1)})+ \cdots + \nu \mathbf{A}_j^{(k)}(\mathbf{I}_T - \nu \mathbf{H}^{(k-1)})\cdots(\mathbf{I}_T - \nu \mathbf{H}^{(1)})\right].
\end{equation}
\Cref{eq:phi jk2} is a matrix form and it holds for all $j=1,\cdots,d$ and any iteration step $k$, regardless of whether the $j$th variable (and its lags) is selected at boosting step $k$. Recall the definition of $\mathbf{A}_{j}^{(k)}$ in \cref{eq:group Ajk}. If the $j$th variable is never selected up to the iteration step $k$, we have $\mathbf{A}_{j}^{(1)} = \cdots = \mathbf{A}_{j}^{(k)}=0$ and $\hat{\phi}_{(j)}^{(k)\prime}=0$; if the $j$th variable is selected in step $k-1$ but not $k$, we have $\mathbf{A}_{j}^{(k-1)} \neq 0$ and $\mathbf{A}_{j}^{(k)} = 0$.

For LS-Boost1, we can use $\mathbf{X}_g$ in \cref{eq:Xg} to rewrite
\begin{align} \label{eq:phi jk3}
	\hat{\phi}_{(j)}^{(k)\prime} &= \tilde{\mathbf{A}}_{j}^{(k)} \left[\mathbf{X}_{(1)} \phi_{(1)}'+\cdots+\mathbf{X}_{(d)} \phi_{(d)}' + \mathbf{u}\right] \nonumber \\
	&=\tilde{\mathbf{A}}_{j}^{(k)} \mathbf{X}_{(1)} \phi_{(1)}'+ \cdots + \tilde{\mathbf{A}}_{j}^{(k)} \mathbf{X}_{(d)} \phi_{(d)}' + \tilde{\mathbf{A}}_{j}^{(k)} \mathbf{u} 
\end{align}
It can be shown that terms such as $\tilde{\mathbf{A}}_{j}^{(k)} \mathbf{X}_{(1)} \phi_{(1)}'$ will converge to a constant as $T \rightarrow \infty$, and we only need to consider the term $\tilde{\mathbf{A}}_{j}^{(k)} \mathbf{u}$ when computing the variance of the elements in $\hat{\phi}_{(j)}^{(k)\prime}$. In fact, we will derive the asymptotic distribution of $\hat{\phi}_{(j)}^{(k)\prime}$, which shows the standard error is asymptotically valid.

We make the following assumptions. 
\begin{assumption} \label{assumption:stationary}
	 $y_t$ is a VAR($p$) process defined in \cref{eq:varp}  and is strictly stationary. 
\end{assumption}
\begin{assumption} \label{assumption:error}
	The errors $\left\{u_t\right\}$ are i.i.d. with $E(u_t) = 0$ and $\text{Var}(u_t) = \Omega$. The errors have bounded fourth moment, $E(u_{i_1t}u_{i_2t}u_{i_3t}u_{i_4t}) < \infty$ for all $i_1,i_2,i_3,i_4=1,\cdots,d$. 
\end{assumption}
\Cref{assumption:stationary,assumption:error} are standard in VAR modeling. In \Cref{assumption:stationary}, we implicitly assume the lag order $p$ is fixed. In practice, $p$ is usually a small number and can be easily tuned or selected by an information criterion. The bounded fourth moment condition in \Cref{assumption:error} implies that the process $\{y_t\}$ itself has a bounded fourth moment, a required condition to apply the central limit theory for a vector martingale difference sequence. If one is only interested in applying the LS-Boost algorithm, the i.i.d. assumption can be relaxed; we only need uncorrelated errors in a VAR($p$) model and the errors may be heteroskedastic or conditional heteroskedastic. The i.i.d. assumption are imposed for the convenience of deriving the standard errors.

Let $\text{vec}$ be the matrix vectorization operator and $\otimes$ be the Kronecker product.
\begin{theorem} \label{thm:group boosting asymp}
	Under \Cref{assumption:stationary,assumption:error}, as $T \rightarrow \infty$, for the selected variable $j$ at boosting step $k \geq 1$,
	\begin{equation} \label{eq:lsbg_asymp}
		\sqrt{T}(\text{vec}(\hat{\phi}_{(j)}^{(k)\prime}) - c_{(j)}^{(k)}) \rightarrow N\left(\mathbf{0}, \Omega^{(k)} \otimes \mathbf{Q}_{(j)}^{(k)}\right),
	\end{equation}
	where $c_{(j)}^{(k)}$, $\Omega^{(k)}$, and $\mathbf{Q}_{(j)}^{(k)}$ are constants that are unique to the $j$th variable at boosting step $k$.
\end{theorem}
The explicit expressions of both $c_{(j)}^{(k)}$ and $\mathbf{Q}_{(j)}^{(k)}$ are given in the proof in the online supplement. All quantities such as $c_{(j)}^{(k)}$, $\Omega^{(k)}$, and $\mathbf{Q}_{(j)}^{(k)}$ are constants for a given $k$ and they do not depend on the sample data (see the proof in \cref{eq:phi term1,eq:phi term1 proof,eq:phi jk5,eq:phi jk5 cjk term,eq:Omega and Q}). \Cref{thm:group boosting asymp} gives a familiar form of the asymptotic result for a VAR that can be found in standard textbooks such as \cite{hamilton1994time,lutkepohl2005time}. What is new here are the expressions for $c_{(j)}^{(k)}$, $\Omega^{(k)}$, and $\mathbf{Q}_{(j)}^{(k)}$. These constants adapt to the iterative nature of the boosting algorithm, and $\mathbf{Q}_{(j)}^{(k)}$ keeps accumulating the statistical uncertainty as the boosting procedure progresses. The diagonal elements of $\hat{\Omega}^{(k)} \otimes \hat{\mathbf{Q}}_{(j)}^{(k)}$ will allow us to construct the \textit{t} statistic and compute the asymptotically valid \textit{p}-value.  In writing \cref{eq:lsbg_asymp}, we assume the \textit{j}th variable is selected as least once up to step \textit{k}; otherwise, $ \phi_{(j)}^{(k)} $ is not estimated and its value is simply $ \mathbf{0} $.

\begin{remark}
	It is clear from \Cref{thm:group boosting asymp} that, for a given $k$, the boosting estimate $\hat{\phi}_{(j)}^{(k)}$ is biased and inconsistent since $c_{(j)}^{(k)} \neq 0$. In high-dimensional modeling, working with a biased estimate in practice is quite common for many adaptive methods such as the lasso and LS-boosting. The new approach this paper adopts is to use the standard error in \cref{eq:lsbg_asymp} to construct the t-statistic for a biased estimate. Since obtaining a full and unbiased LS solution is practically impossible, working with a biased estimate for hypothesis testing provides an alternative. Consequently, the hypothesis is not for the population parameter $H_0{:}\;\text{vec}(\phi_{(j)}') = \boldmath{0}$. Instead, we test the statistical significance of a LS boosting estimate up to every step $k$. See Section 5 for further discussions.
\end{remark}

\begin{remark}
	In estimation, $\hat{\Omega}^{(k)}$ can be obtained from the cross product of the residuals ($\hat{\mathbf{u}}^{(k)}$). When the dimension $d$ is large, this estimate becomes singular and one has to introduce sparsity in $\Omega^{(k)}$ to get a feasible estimate. There is a large body of research on estimating high-dimensional variance covariance matrix. Notice that we only need  to estimate the diagonal elements of $\hat{\Omega}^{(k)} \otimes \hat{\mathbf{Q}}_{(j)}^{(k)}$ for standard errors, but not the entire matrix. This greatly reduces the number of parameters and makes our procedure feasible for high-dimensional VARs.
\end{remark}

To gain more insights into \Cref{thm:group boosting asymp}, let us consider the case at boosting step $k=1$ when the $j_1$th variable is selected ($j=j_1$). When $k = 1$, we have
\begin{align*}
	\hat{\phi}_{(j_1)}^{(1)\prime} &= \nu \mathbf{A}_{j_1}^{(1)} \mathbf{Y}\\
	&= \nu (\mathbf{X}_{(j_1)}'\mathbf{X}_{(j_1)})^{-1}\mathbf{X}_{(j_1)}' \mathbf{Y}\\
	&=\nu (\mathbf{X}_{(j_1)}'\mathbf{X}_{(j_1)})^{-1}\mathbf{X}_{(j_1)}' (\mathbf{X}_{(1)} \phi_{(1)}' + \cdots + \mathbf{X}_{(j_1)} \phi_{(j_1)}' + \cdots + \mathbf{X}_{(d)} \phi_{(d)}'+ \mathbf{u})\\
	&= \nu \left(\frac{\mathbf{X}_{(j_1)}'\mathbf{X}_{(j_1)}}{T}\right)^{-1} \left(\frac{\mathbf{X}_{(j_1)}' \mathbf{X}_{(1)}}{T}\right) \phi_{(1)}' + \cdots + \nu \phi_{(j_1)}' + \cdots \\
	&\quad + \nu \left(\frac{\mathbf{X}_{(j_1)}'\mathbf{X}_{(j_1)}}{T}\right)^{-1} \left(\frac{\mathbf{X}_{(j_1)}' \mathbf{X}_{(d)}}{T}\right) \phi_{(d)}' + \nu (\mathbf{X}_{(j_1)}'\mathbf{X}_{(j_1)})^{-1}\mathbf{X}_{(j_1)}' \mathbf{u}.
\end{align*}

Under \Cref{assumption:stationary,assumption:error}, all cross product terms can be shown to converge to a constant as $T \rightarrow \infty$, and we have
\begin{equation} \label{eq:phij11}
	\hat{\phi}_{(j_1)}^{(1)\prime} - c_{(j_1)}^{(1)} = \nu \left(\frac{\mathbf{X}_{(j_1)}'\mathbf{X}_{(j_1)}}{T}\right)^{-1} \frac{1}{T} \sum_{t=1}^{T} \mathbf{X}_{(j_1),t}' u_t,
\end{equation}
where $\mathbf{X}_{(j_1),t}'$ is the $t$th column of $\mathbf{X}_{(j_1)}'$, $u_t$ is the $t$th row of $\mathbf{u}$, and 
\begin{align} 
	c_{(j_1)}^{(1)} &= \lim\limits_{T \rightarrow \infty} \nu \left(\frac{\mathbf{X}_{(j_1)}'\mathbf{X}_{(j_1)}}{T}\right)^{-1} \left(\frac{\mathbf{X}_{(j_1)}' \mathbf{X}_{(1)}}{T}\right) \phi_{(1)}' + \cdots + \nu \phi_{(j_1)}' + \cdots \nonumber \\
	&\quad+ \nu \left(\frac{\mathbf{X}_{(j_1)}'\mathbf{X}_{(j_1)}}{T}\right)^{-1} \left(\frac{\mathbf{X}_{(j_1)}' \mathbf{X}_{(d)}}{T}\right) \phi_{(d)}'. \label{eq:c11}
\end{align}
Let the limit of $ \nu \left(\frac{\mathbf{X}_{(j_1)}'\mathbf{X}_{(j_1)}}{T}\right)^{-1}$ be $c$. Vectorizing \cref{eq:phij11} gives
\begin{equation} \label{eq:phij11 vec}
	\begin{bmatrix}
	\hat{\phi}_{(j_1),1}^{(1)\prime} - c_{(j_1),1}^{(1)}\\
	\vdots\\
	\hat{\phi}_{(j_1),i}^{(1)\prime} - c_{(j_1),i}^{(1)}\\
	\vdots\\
	\hat{\phi}_{(j_1),d}^{(1)\prime} - c_{(j_1),d}^{(1)}
	\end{bmatrix}
	=
	\begin{bmatrix}
	\frac{1}{T} \sum_{t=1}^{T} c\mathbf{X}_{(j_1),t}' u_{t_1}\\
	\vdots\\
	\frac{1}{T} \sum_{t=1}^{T} c\mathbf{X}_{(j_1),t}' u_{t_i}\\
	\vdots\\
	\frac{1}{T} \sum_{t=1}^{T} c\mathbf{X}_{(j_1),t}' u_{t_d}
	\end{bmatrix}_{dp \times 1},
\end{equation}
where $\hat{\phi}_{(j_1),i}^{(1)\prime}$ and $c_{(j_1),i}^{(1)}$ are the $i$th column of $\hat{\phi}_{(j_1)}^{(1)\prime}$ and $c_{(j_1)}^{(1)}$, respectively; $u_{t_i}$ is the $i$th element of the $1 \times d$ row vector $u_t$. \Cref{eq:phij11 vec} gives a convenient form to derive the asymptotic distribution. To compute the (asymptotic) standard error, consider the $i$th $p \times 1$ vector $\hat{\phi}_{(j_1),i}^{(1)\prime} - c_{(j_1),i}^{(1)}$ on the l.h.s. of \cref{eq:phij11 vec}. It can be shown that $E(u_{t_i}|\mathbf{X}_{(j_1),t}') = 0$, because $u_{t_i}$ is the error at time $t$ while $\mathbf{X}_{(j_1),t}'$ includes $p$ lags before time $t$. 
\begin{align*}
	\text{Var}(\hat{\phi}_{(j_1),i}^{(1)\prime} - c_{(j_1),i}^{(1)}) &= \text{Var}\left(\frac{1}{T} \sum_{t=1}^{T} c\mathbf{X}_{(j_1),t}' u_{t_i}\right)\\
	&= \frac{1}{T} \text{Var}(c\mathbf{X}_{(j_1),t}' u_{t_i}) \qquad \qquad\text{ because $u_t$ is i.i.d.}\\
	&= \frac{1}{T} E(c\mathbf{X}_{(j_1),t}' u_{t_i}^2 \mathbf{X}_{(j_1),t}c') \qquad \text{because $E(u_{t_i}|\mathbf{X}_{(j_1),t}') = 0$}\\
	&=\frac{1}{T} c E(\mathbf{X}_{(j_1),t}' \mathbf{X}_{(j_1),t})c' \sigma_i^2,
\end{align*}
where we replace $c$ with $\nu (\frac{\mathbf{X}_{(j_1)}'\mathbf{X}_{(j_1)}}{T})$, $E(\mathbf{X}_{(j_1),t}' \mathbf{X}_{(j_1),t})$ with $\mathbf{X}_{(j_1)}' \mathbf{X}_{(j_1)}/T$ in estimation. $\sigma_i^2$ is the $i$th diagonal element of $\Omega^{(1)}$ and can be estimated based on the residuals at step $ k $.

When $k > 1$, the computation for the standard error becomes more involved but essentially follows the same steps.

\begin{remark}
	From the definition of $c_{(j_1)}^{(1)}$ in \cref{eq:c11}, it is clear that the parameter estimator $\hat{\phi}_{(j_1)}^{(1)\prime}$ in \cref{eq:phij11} is biased, which is expected for an iterative procedure like the LS-Boost. The key result is, despite of the bias, \cref{eq:phij11} allows us to compute the asymptotic variance of $\hat{\phi}_{(j_1)}^{(1)\prime}$ at step $k=1$, and this asymptotic variance characterizes the statistical uncertainty in estimating $\hat{\phi}_{(j_1)}^{(1)\prime}$.
	
\end{remark}

Next, we use a special case to illustrate the behavior of $\hat{\phi}_{(j_1)}^{(1)\prime}$ when $k$ is large. Consider the expression for $\hat{\phi}_{(j)}^{k\prime}$ in \cref{eq:phi jk2}. Assume variable $j$ is always chosen throughout all $k$ steps so that $\mathbf{A}_j^{(1)}=\cdots=\mathbf{A}_j^{(k)} = \mathbf{A}_j$ and $\mathbf{H}^{(1)} = \cdots = \mathbf{H}^{(k)} = \mathbf{H}$, and \cref{eq:phi jk2} becomes

\begin{align} \label{eq:phi jk7}
\hat{\phi}_{(j)}^{(k)\prime} &=\nu \mathbf{A}_j \mathbf{Y} + \nu \mathbf{A}_j(\mathbf{I}_T - \nu \mathbf{H}) \mathbf{Y}+\nu \mathbf{A}_j(\mathbf{I}_T - \nu \mathbf{H})(\mathbf{I}_T - \nu \mathbf{H}) \mathbf{Y} \nonumber \\
&\quad + \cdots + \nu \mathbf{A}_j(\mathbf{I}_T - \nu \mathbf{H})\cdots(\mathbf{I}_T - \nu \mathbf{H}) \mathbf{Y}  \nonumber\\
&=\nu \mathbf{A}_j \left[ (\mathbf{I}_T - \nu \mathbf{H})^0 + (\mathbf{I}_T - \nu \mathbf{H})^1 + \cdots + (\mathbf{I}_T - \nu \mathbf{H})^{(k-1)} \right] \mathbf{Y} \nonumber\\
&\rightarrow \nu \mathbf{A}_j  \left[\mathbf{I}_T - (\mathbf{I}_T - \nu\mathbf{H})\right]^{-1} \mathbf{Y}   \quad\text{  as } k \rightarrow \infty \nonumber\\
&= \mathbf{A}_j \mathbf{Y} =  (\mathbf{X}_{(j)}'\mathbf{X}_{(j)}^{-1} \mathbf{X}_{(j)}'\mathbf{Y}.
\end{align}
Since $\mathbf{H}$ is idempotent with eigenvalues $0$ and $1$, Theorem 4.3.1 in \cite{hornjohnson1985matrixanalysis} to show the eigenvalues of $\mathbf{I}_T-\nu \mathbf{H}$ is between $0$ and $1$. Diagonalization gives $ (\mathbf{I}_T - \nu \mathbf{H})^k = \mathbf{U} \mathbf{D}_k \mathbf{U}'$, where $ \mathbf{U} $ is orthonormal and $ \mathbf{D}_k $ is diagonal with all eigenvalues of $ (\mathbf{I}_T - \nu \mathbf{H})^k\ $, which will allows us to derive a bias expression similar to that in Proposition 3 and Theorem 1 of \cite{buhlmannandyu2003JASAL2boosting} for the case of multivariate regression. Since $ \hat{\phi}_{(j)}^{(k)\prime}  $ converges to the LS estimator, it is unbiased as $k \rightarrow \infty$ and will have the same variance as the LS estimator.

In practice, we will have many variables and the boosting algorithm will not update the same variable at every step so that some (possibly many) of the $\mathbf{A}_j^{(k)}$ in \cref{eq:phi jk2} are zero, But it can be roughly seen that, as long as the boosting step $k$ is large enough, coefficient matrix for variable $j$ will get many updates. These incremental updates constitute a subsequence of the geometric matrix series in \cref{eq:phi jk2} and will also converge to the same limit. In \Cref{sec:computation_bounds}, we provide a different perspective on the convergence property of the LS-Boost estimator when $k$ is large.


\subsection{Asymptotic results for LS-Boost2 when $k$ is fixed} \label{sec:asym_lsboost2}
Next we discuss the asymptotic results for the estimator described in \Cref{sec:second alg}, and it will be similar to \Cref{thm:group boosting asymp}.

By the definition of $\hat{\beta}_{(j)s}^{(k)\prime}$ in \cref{eq:group boost obj js}, we have
\begin{align} \label{eq:phi jk1 js}
\hat{\phi}_{(j)s}^{(k)\prime} &= \hat{\phi}_{(j)s}^{(k-1)\prime} + \nu \hat{\beta}_{(j)s}^{(k)\prime} \nonumber \\
&= \nu \hat{\beta}_{(j)s}^{(1)\prime} + \cdots + \nu \hat{\beta}_{(j)s}^{(k)\prime} \nonumber \\
&= \tilde{\mathbf{A}}_{js}^{(k)} \mathbf{Y}, \text{ similar to }\cref{eq:phi jk2}
\end{align}
where
\begin{equation} \label{eq: A tilde jks}
\tilde{\mathbf{A}}_{js}^{(k)} = \left[\nu \mathbf{A}_{js}^{(1)} + \nu \mathbf{A}_{js}^{(2)}(\mathbf{I}_T - \nu \mathbf{H}^{(1)})+ \cdots + \nu \mathbf{A}_{js}^{(k)}(\mathbf{I}_T - \nu \mathbf{H}^{(k-1)})\cdots(\mathbf{I}_T - \nu \mathbf{H}^{(1)})\right].
\end{equation}

Hence, similar to \cref{eq:phi jk3}, we have
\begin{align} \label{eq:phi jk3 js}
\hat{\phi}_{(j)s}^{(k)\prime} 
&=\tilde{\mathbf{A}}_{js}^{(k)} \mathbf{X}_{(1)s} \phi_{(1)s}'+ \cdots + \tilde{\mathbf{A}}_{js}^{(k)} \mathbf{X}_{(d)s} \phi_{(d)s}' + \tilde{\mathbf{A}}_{js}^{(k)} \mathbf{u}.
\end{align}

The derivation of the asymptotic distribution result largely follows the proof of \Cref{thm:group boosting asymp}, and we give the result as a corollary in the following.
\begin{corollary} \label{cor:boost_asymp}
	Under \Cref{assumption:stationary,assumption:error}, as $T \rightarrow \infty$, for every selected variable $j$ and its lag $s$ at every boosting step $k \geq 1$,
	\begin{equation}
	\sqrt{T}(\hat{\phi}_{(j)s}^{(k)} - c_{(j)s}^{(k)}) \rightarrow N\left(\mathbf{0}, \Omega^{(k)} \cdot  \mathbf{Q}_{(j)s}^{(k)}\right),
	\end{equation}
	where $c_{(j)s}^{(k)}$, $\Omega^{(k)}$, and $\mathbf{Q}_{(j)s}^{(k)}$ are constant that are unique to the $j$th variable and its  $s$th lag at boosting step $k$.
\end{corollary}
See the supplement for the proof. Note that $\mathbf{Q}_{(j)s}^{(k)}$ is a scalar in the above corollary. Again we assume the parameter $\phi_{(j)s}$ gets updated at least once between steps $1$ and $k$.

\section{Convergence of LS-Boost1 estimator for fixed-\textit{T}} \label{sec:computation_bounds}

While \Cref{thm:group boosting asymp} and \Cref{cor:boost_asymp} characterize the behavior of the estimator when $k$ is fixed and $T \rightarrow \infty$, it will also be useful to study its behavior when $k \rightarrow \infty$ and $T$ is fixed. At the end of \Cref{sec:asym_lsboost1}, we briefly discuss the general case of a VAR with many variables and possibly multiple lags. The LS-boost procedure will likely select a different variable at each step, and the expression for $\tilde{\mathbf{A}}_{(j)}^{(k)}$ in \cref{eq: A tilde jk} becomes more complicated. To see the limit of the estimator when $k \rightarrow \infty$ and $T$ is fixed , we instead derive a computation bound result for the VAR LS-Boost estimator, similar to the cross-section result in \cite{freundetal2017boosting}. The following result is non-asymptotic, but it can characterize the behavior of the estimator as $k \rightarrow \infty$.

Consider the VAR model in \cref{eq:varp_grped}. Let $\lambda_{\text{pmin}}(\mathbf{X}_g'\mathbf{X}_g)$ be the smallest non-zero eigenvalue of $\mathbf{X}_g'\mathbf{X}_g$ and define the linear convergence rate coefficient
\begin{equation} \label{eq:linear rate}
	\gamma = 1 - \frac{\nu(2-\nu)\lambda_{\text{pmin}}(\mathbf{X}_g'\mathbf{X}_g)}{4d}.
\end{equation}
\cite{freundetal2017boosting} shows that $0.75 \leq \gamma < 1$ when columns of $\mathbf{X}_g$ are normalized with unit $\ell_2$ norm.  To handle time series correlation in a VAR, we further assume the data $ \mathbf{X}_{(j)} $ are re-scaled so that its inner product is an identity matrix. This normalization is only used for the convenience of deriving the theoretical results, and it is not needed for the actual implementation of the boosting algorithms. More details are provided in the proof of the following theorem. Let the LS solution at boosting step $k$ be $\bm{\phi}^{(k)}_{g,\text{LS}}$ and let its estimator be $\hat{\bm{\phi}}^{(k)}_{g,\text{LS}}$. Also let $ \hat{\bm{\phi}}_{g,\text{LS}} $ be the LS estimator for the VAR model in  \cref{eq:varp_grped}.

\begin{theorem} \label{thm:computational bounds}
	For the LS-Boost1 algorithm described in \Cref{sec:group boosting}, under \Cref{assumption:stationary} with i.i.d. errors, the estimated coefficient bound is given by
	\begin{equation} \label{eq:coef_bound}
		\lVert \hat{\bm{\phi}}_g^{(k)} - \hat{\bm{\phi}}^{(k)}_{g,\text{LS}} \rVert_{2} \leq \frac{\lVert \mathbf{X}_g \hat{\bm{\phi}}_{\text{g,LS}} \rVert_2^2}{\lambda_{\text{pmin}}(\mathbf{X}_g'\mathbf{X}_g)} \gamma^{k/2}.
	\end{equation}
	The prediction bound is given by
	\begin{equation} \label{eq:prediction_bound}
		\lVert \mathbf{X}_g \hat{\bm{\phi}}_g^{(k)} - \mathbf{X}_g \hat{\bm{\phi}}_{g,\text{LS}} \rVert_2 \leq \lVert \boldsymbol{\mathrm{X}} \hat{\bm{\phi}}_\text{g,LS} \rVert_2 \gamma^{k/2} .
	\end{equation}
\end{theorem} 
See the supplement for the proof. \Cref{thm:computational bounds} is similar to Theorem 2.1 in \cite{freundetal2017boosting}. Our results differ from theirs in two aspects: our results are derived for a VAR and we need to take care of additional dependence in time series. Similar to the definition of $\bm{\phi}_{g}$ below \cref{eq:varp_grped}, $\hat{\bm{\phi}}_g^{(k)}$ in \cref{eq:coef_bound} is defined as $\hat{\bm{\phi}}_g^{(k)} = \left[\hat{\phi}_{(1)}^{(k)}, \cdots, \hat{\phi}_{(d)}^{(k)}\right]'$.

Notice that results in \cref{eq:coef_bound,eq:prediction_bound} require very few assumptions to hold. They are derived for the fixed-$T$ case, and moment conditions in \Cref{assumption:error} are not needed here.

Normally, we will simply write $\bm{\phi}^{(k)}_{g}$ instead of $\bm{\phi}^{(k)}_{g,\text{LS}}$ for the unknown parameter. The subscript ``LS" in $\bm{\phi}^{(k)}_{g,\text{LS}}$ is added to emphasize that, at boosting step $k$ and corresponding to the specific sparsity at step $k$, there exist an (or multiple) unknown LS solution $\bm{\phi}^{(k)}_{g,\text{LS}}$.     \Cref{eq:coef_bound} states that $\hat{\bm{\phi}}_g^{(k)}$ converges to the closest (in $\ell_2$ norm) LS estimate . Notice that $\bm{\phi}^{(k)}_{g,\text{LS}}$ changes as $k$ changes. Depending on $k$,  $\bm{\phi}^{(k)}_{g,\text{LS}}$ may or may not be unique. When $\bm{\phi}^{(k)}_{g,\text{LS}}$ is non-unique, $\hat{\bm{\phi}}^{(k)}_{g,\text{LS}}$ estimates one of the LS solutions at step $k$ and $\hat{\bm{\phi}}_g^{(k)}$ converges to one of the many $\bm{\phi}^{(k)}_{g,\text{LS}}$s. To see why the limit $\bm{\phi}^{(k)}_{g,\text{LS}}$ may change with $k$, consider the following two scenarios: When $k$ is small, LS-Boost1 generates very few nonzeros in the coefficient matrices and a LS solution for these nonzeros is unique for the model at step $k$; when $k$ is large, LS-Boost1 may have generated a large number of nonzeros in the coefficient matrices and a LS solution for these nonzeros exists but is non-unique. In the case of prediction, \Cref{eq:prediction_bound} directly shows that the LS-Boost1 prediction converges to the LS prediction at a linear rate. Note that, by using an $\alpha$-mixing condition for the data and several other assumptions, Theorem 2 in \cite{lutz2006boosting} obtains a result for the consistency of LS-Boost prediction when $T$ and $k$ go to infinity simultaneously, whereas \Cref{thm:group boosting asymp,thm:computational bounds} in this paper permits the computation of the s.e. and can characterize precisely the convergence of LS-Boost prediction as a function of the boosting step $k$.

From \Cref{thm:computational bounds}, we can immediately deduce the behavior of $\hat{\bm{\phi}}_g^{(k)}$ when $k \rightarrow \infty$ with a fixed $T$, and \Cref{thm:group boosting asymp} gives the result when $T \rightarrow \infty$ with a fixed $k$. What would happen if both $k$ and $T \rightarrow \infty $? The result will depend on the relative speed at which $dp$, $T$ and $k \rightarrow \infty$. Loosely speaking, if $T > dp$, the estimator from LS-Boost1 will converge to the unique LS solution and to one of the non-unique LS solutions if $T < dp$ as $k \rightarrow \infty$. The final result will also be influenced by the interplay of all these factors and the actual sparsity of the coefficient matrices.

Our discussion in this section focuses on LS-Boost1 with the understanding the LS-Boost2 is a special case of LS-Boost1.

\section{Additional discussions on the \textit{p}-value} \label{sec:discussions}

We add some additional discussions on the estimator and its \textit{p}-value.

\textit{Incremental hypothesis testing}. It is mentioned in \Cref{sec:intro} that the incremental hypothesis studied in this paper is not equivalent to the classical hypothesis for the population regression parameters. Consider the $k=1$ example discussed in \Cref{sec:asym_lsboost1}. \Cref{thm:group boosting asymp} implies that when $k=1$, as $T \rightarrow \infty$, $\text{vec}(\hat{\phi}_{(j)}^{(k)\prime})$ converges to $c_{(j_1)}^{(1)}$ in \cref{eq:c11} and $c_{(j_1)}^{(1)}$ can be viewed as the vectorized version of $\phi_{(j_1),\text{LS}}^{(1)}$, the $j_1$th block of $\bm{\phi}_{\text{g,LS}}^{(1)}$ in \Cref{thm:computational bounds}. Hence, for a fixed boosting step $k$, the true unknown value of $\hat{\phi}_{(j)}^{(k)\prime}$ is $\phi_{\text{(j),LS}}^{(k)}$, which is different from the LS solution to \cref{eq:varp_grped}, $\phi_{(j),\text{LS}}$. At each boosting step, the null hypothesis of the \textit{t}-test for each element of the coefficient matrix of the $j$th variable is 
\begin{equation} \label{eq:null}
	H_0{:}\;\phi_{\text{(j),LS},rq}^{(k)} = 0, \text{ for }r=1,\cdots,d \text{ and }q = 1,\cdots,p,
\end{equation}
where $\phi_{\text{(j),LS},rq}^{(k)}$ is the $rq$th element of the matrix $\phi_{\text{(j),LS}}^{(k)}$, whose transpose, $\phi_{\text{(j),LS}}^{(k)\prime}$, is the $j$th $p \times d$ block of $\bm{\phi}_{\text{g,LS}}^{(k)}$, and the value of $\phi_{\text{(j),LS}}^{(k)}$ may change with $k$. The LS-Boost1 estimator $\hat{\phi}_{(j)}^{(k)}$ is an estimator for $\phi_{\text{(j),LS}}^{(k)}$ and it reflects the cumulative incremental changes in the $j$th coefficient matrix. Hence, the \textit{t}-test and its \textit{p}-value repeatedly test the null in \cref{eq:null} for each nonzero entry in $\hat{\phi}_{(j)}^{(k)}$ at every boosting step. If any element in $\hat{\phi}_{(j)}^{(k)}$ is $0$, we do not have a boosting estimate yet and cannot perform a \textit{t}-test for the corresponding element in $\phi_{\text{(j),LS}}^{(k)}$.

\textit{Interpretation of an incremental \textit{p}-value}. Note that  $\phi_{\text{(j),LS},rq}^{(k)} \neq \phi_{\text{(j),LS},rq}$, i.e., the null in \cref{eq:null} is not about the parameter in the population regression equation. Thus, when rejecting the null in \cref{eq:null} at $5\%$ significance level, we conclude that $\phi_{\text{(j),LS},rq}^{(k)}$ is statistically significant but with no implication on the significance of $\phi_{\text{(j),LS},rq}$. Alternatively, since the \textit{t}-statistic is
\begin{equation} \label{eq:t-stat}
    t\text{-statistic} = \frac{\hat{\phi}_{(j),rq}^{(k)} - 0}{\sqrt{\widehat{\text{var}}(\hat{\phi}_{(j),rq}^{(k)})}},
\end{equation}
we conclude that, given the model up to step $k$, the cumulative incremental change, $\hat{\phi}_{(j),rq}^{(k)}$, is statistically different from $0$ when the null is rejected. The denominator of \cref{eq:t-stat} is obtained from the asymptotic distribution in \Cref{thm:group boosting asymp}. If a parameter estimate is $0.1$ and is significant at step $k$, when it becomes $0.2$ at step $k+1$, is it also significant at step $k+1$? Not necessary. \Cref{fig:plot_pval} shows some examples where estimates are significant when $k$ is small but become insignificant when $k$ is large. Through repeated hypothesis testing we can track the significance of a parameter as it evolves along the path of the boosting algorithm. Our s.e. and \textit{p}-value are computed by assuming the model up to step $k$ is given and they do not quantify the statistical uncertainty associated with model selection at step $k$. Explicitly incorporating the model selection uncertainty can sometimes lead to wide confidence intervals, see examples for cross-section regression in Figure 6.12 in \cite{hastie2015slsparcity} and Figure 1 in \cite{tibshiranietal2016lassopval}. 


Since $\phi_{\text{(j),LS},rq}^{(k)}$ in \cref{eq:t-stat} changes with $k$, it is clear that our null hypothesis also changes at each boosting step when LS-Boost adds a new variable in the current step. Hence, the null in our test is not as \textit{fixed} as the classical null of $H_0{:}\;\phi_{\text{(j)},rq} = 0$, but they are similar: In classical hypothesis testing, given all the variables and a linear model, we test if $\phi_{\text{(j)},rq}=0$ without questioning where does the model (or all the variables) come from; in the case of LS-Boost, we test, given the model at a boosting step, if $\phi_{\text{(j),LS},rq}^{(k)} = 0$. 

The superscript $(k)$ makes the null hypothesis in \cref{eq:null} a little unusual. But it is no more unusual than some of the common practices in statistics. We draw an analogy between our incremental hypothesis testing and the use of many shrinkage methods in linear regression. Let $\phi_{(j),rq}$ be a true unknown parameter in a linear regression, say, VAR. A shrinkage method such as the lasso or LS-Boost suggests that, even when the LS estimate is available, it is better not to use it in prediction because a LS solution only minimizes the error on the training data but not the error on the test data (\textit{test error}); instead, use c.v. or AIC to screen the solution path spawn by the algorithm to get a biased estimate of $\phi_{(j),rq}$, which typically gives better prediction result and can minimize the test error. It would seem unusual not to use a LS estimate, but it reflects the practice of bias-variance trade-off in statistics. Similarly in hypothesis testing, when testing the significance of a coefficient, we propose not to use the full LS solution --- even when such a LS solution is unique and/or computable. Instead, use the sequential s.e. and \textit{p}-value to inspect when a fraction of $\phi_{(j),rq}$ is significant and when it is not. Hence, it is proper to use a superscript so that it becomes $\phi_{(j),rq}^{(k)}$. The purpose of sequential hypothesis testing is similar to the use of c.v. in model selection for the lasso and many other shrinkage methods: to uncover better solutions in spite of their bias.

\textit{Proper use of the \textit{p}-value}. We give two examples of proper use of the proposed \textit{p}-value. Example $1$: Given the LS-Boost solution paths of many parameters, an analyst can choose any boosting step $k$ and use the \textit{p}-value to inspect the statistical significance of a selected parameter with a possible Bonferroni correction. Confidence intervals can be constructed based on the associated standard errors. Example 2: Given the LS-Boost solution paths of many parameters, the analyst can use the \textit{p}-values to remove non-significant parameters at each boosting step and apply c.v. or validation procedure for model selection afterwards. Since the \textit{p}-value is used before c.v., the statistical uncertainty associated with c.v. will not affect the validity of the \textit{p}-value. An incorrect application of the \textit{p}-value is to use c.v. first to select a model then apply the \textit{p}-values to further pare down parameters. It is incorrect because the \textit{p}-values are computed at the same time when the boosting algorithm generates the solution path; using c.v. incurs additional uncertainty that are not reflected in the \textit{p}-values so that the \textit{p}-values cannot be used unless such additional uncertainty is small.

\textit{Trade-off between the false positive rate and the false negative rate}. Our simulation results suggest that we can use a $p$-value to keep the FPR of a model under control at the cost of increasing the FNR (FPR and FNR are defined in \Cref{sec:simulation_setup}). A data analyst has to decide whether FPR or FNR is of the most concern. Typically, it is FPR. An equally important topic is proper control of the false discovery rate (FDR). In a time series model like VAR, the \textit{p}-values are dependent not only across variables at each boosting step but also across boosting steps for each variable. It will be very useful to design a method to incorporate the standard Benjamini-Hochberg procedure for FDR control in high-dimensional VARs. This topic is beyond the scope of the current paper.

\textit{A choice between prediction and model interpretation}. When a \textit{p}-value is supplied, one generally expects it will improve the model. The specific nature of such improvement depends on the goal of research. Removing a variable based on its \textit{p}-value does not necessarily improve prediction. We provide some mixed results in the simulation exercise. Hence, it is understood that the goal leans more towards better model interpretation when a \textit{p}-value is applied. Our application, though, does provide an example where a \textit{p}-value-adjusted boosting model improves both prediction and model interpretation over an unadjusted boosting model.

\section{Simulation results}
We conduct Monte Carlo simulation in this section to study the use of the proposed \textit{p}-values in high-dimensional VAR modeling. We first consider a small model and compare the boosted \textit{p}-values to those obtained from the LS method, and then move onto a larger scale simulation study with high-dimensional VARs.

\subsection{A simple example of bivariate VAR} \label{sec:bivariate_simulation}
Let us consider the following stationary bivariate VAR selected from equation (4.2.1) in \cite{lutkepohl2005time}.
\begin{equation} \label{eq:bivariate_var}
	y_t = \begin{bmatrix}
	0.02 \\
	0.03
	\end{bmatrix} + \begin{bmatrix}
	0.5 & 0.1\\
	0.4 & 0.5
	\end{bmatrix} y_{t-1} +
	\begin{bmatrix}
	0 & 0 \\
	0.25 & 0
	\end{bmatrix} y_{t-2} +
	u_t  \text{ and } u_t \sim N\left(\mathbf{0}, \begin{bmatrix}
	0.09 & 0\\
	0    & 0.04
	\end{bmatrix}\right).
\end{equation}
We simulate a sample of $500$ observations and report in \Cref{tab:bivariate_VAR} the estimates and the \textit{p}-values from the LS method along with those of LS-Boost1 and LS-Boost2. The learning rate $\nu$ is chosen to be $0.1$. The selected stopping step at each replication is determined by the bias-corrected Akaike information criterion (AIC) in \cite{hurvichandtsai1989bcAIC} for the two boosting methods.

\begin{table}[htp] \centering
	\begin{center}
		\caption{Average estimates and \textit{p}-values for the bivariate VAR model} 
		\label{tab:bivariate_VAR} 
		\begin{threeparttable}
			\begin{tabular}{lllllllll}  
				\toprule
				&\multicolumn{2}{c}{$\hat{\phi}_1$} & \multicolumn{2}{c}{$\hat{\phi}_2$} & \multicolumn{2}{c}{\textit{p}-value for $\hat{\phi}_1$} &\multicolumn{2}{c}{\textit{p}-value for $\hat{\phi}_2$} \\
				\cmidrule{1-2} \cmidrule(lr){2-3} \cmidrule(lr){4-5} \cmidrule(lr){6-7} \cmidrule(lr){8-9}
				LS-Boost1 & 0.494 & 0.095 & -0.001 & -0.006 & 0.000 & 0.255 & 0.499 & 0.553 \\
				          & 0.403 & 0.476 & 0.258 & 0.009 & 0.000 & 0.000 & 0.000 & 0.513 \\
				\cmidrule{1-2} \cmidrule(lr){2-3} \cmidrule(lr){4-5} \cmidrule(lr){6-7} \cmidrule(lr){8-9}
				LS-Boost2 & 0.470 & 0.095 & 0.022 & -0.006 & 0.000 & 0.089 & 0.328 & 0.620 \\
				          & 0.419 & 0.478 & 0.234 & 0.010 & 0.000 & 0.000 & 0.000 & 0.465 \\
				\cmidrule{1-2} \cmidrule(lr){2-3} \cmidrule(lr){4-5} \cmidrule(lr){6-7} \cmidrule(lr){8-9}
				LS       & 0.495 & 0.092 & 0.002 & -0.005 & 0.000 & 0.288 & 0.499 & 0.554 \\
				          & 0.403 & 0.490 & 0.247 & 0.005 & 0.000 & 0.000 & 0.000 & 0.521 \\	
				\bottomrule
			\end{tabular}
			\begin{tablenotes}[flushleft]
				\setlength\labelsep{0pt}
				\item[] \textit{Notes}: This table reports estimation output with a sample size of $500$. All numbers are averages over $100$ replications. Results for the intercept are skipped in this table.
			\end{tablenotes}
		\end{threeparttable}
	\end{center} 
\end{table}

In \Cref{tab:bivariate_VAR}, the LS-Boost estimates are very close to those of the LS. More interestingly, the statistical decision based on LS-Boost \textit{p}-values gives the identical results compared to the LS if we use $5\%$ as the critical value. Compared to the result of LS-Boost2, the result of LS-Boost1 is slightly closer to that of the LS, perhaps due to the fact that, on average, columns $\mathbf{X}_g$ in LS-Boost1 is updated more often since LS-Boost1 selects all lags of a single variable.

\begin{figure}[htp]
	\centering   
	\subfloat[LS-Boost1 s.e. and LS s.e.  \label{fig:sebars}]{{\includegraphics[width=0.5\linewidth,keepaspectratio]{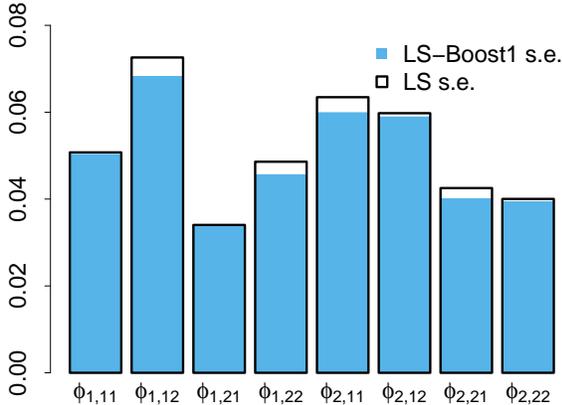} }}%
	\subfloat[Path of LS-Boost1 \textit{p}-values and the LS \textit{p}-value \label{fig:plines}]{{\includegraphics[width=0.5\linewidth,keepaspectratio]{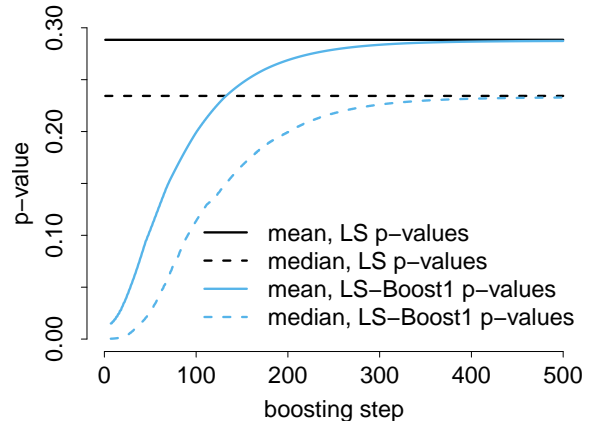} }}%
	\caption{\Cref{fig:sebars} plots the average s.e. of the $8$ coefficient estimators in \cref{eq:bivariate_var} of AIC-selected LS-Boost1 models in $100$ replications. The height of the white bar outlined with a black line is the average s.e. of the corresponding LS in $100$ replications. \Cref{fig:plines} plots the entire path of the mean (over $100$ replications) and the median \textit{p}-value of LS-Boost1. The black solid and dashed lines are the mean and median \textit{p}-values of the LS in $100$ replications.}%
\end{figure}

It is also interesting to observe the boosting \textit{p}-value can be very close to that of the LS. For example, the two \textit{p}-values of LS-Boost1, $0.499$ and $0.553$, are close to $0.499$ and $0.554$ of the LS. \Cref{fig:sebars} plots the average s.e. of $100$ A.I.C.-selected LS-Boost1 models for the $8$ parameters in \cref{eq:bivariate_var} vs. those of the LS. The s.e. of A.I.C.-selected models are close to those of the LS, in spite of the fact A.I.C. typically stops the boosting iteration very early in the simulation. \Cref{fig:plines} offers a more complete picture by plotting the average of LS-Boost1 s.e. for the entire $500$ steps for the parameter $\phi_{1,12}$. It demonstrates that our \textit{p}-value can converge to that of the LS when $k$ is large. This is not a surprise since the sample size ($500$) is large relative to the number of parameters in \cref{eq:bivariate_var}, the bivariate VAR has a unique LS solution and \Cref{thm:computational bounds} predicts the LS-Boost1 estimator converges to this unique LS solution.


\subsection{Simulation results for high-dimensional VARs}
In this section, we consider high-dimensional VARs and focus on using the $5\%$ cutoff to remove variables and study its impact on model performance. This is arguably the most familiar application of a \textit{p}-value.

\subsubsection{Simulation setup} \label{sec:simulation_setup}
We use a VAR(2) model with no intercept in our simulation. Consider the model
\begin{equation} \label{eq:var2_simu}
	y_t = \phi_1 y_{t-1} +  \phi_2 y_{t-2} + u_t,
\end{equation}
where we make the following specifications:

(i) Both $\phi_1$ and $\phi_2$ are $d \times d$ matrices and are generated according to a sparsity parameter $s$. Three combinations of sample size $(T)$, model dimension $(d)$, and coefficient sparsity $(s)$ are considered (see \Cref{tab:model_specs} for more details).

\begin{table}[htp]
	\centering
	\caption{Sample size, model dimension, and sparsity}
	\begin{tabular}{lrrrrrr}
		\toprule
		\text{model type}& $T$  & $d$ & $s$ & No. of nonzeros & No. of parameters & sparsity rate \\
		                 &      &     &     & $d \times s \times 2$ & $d \times d \times 2$ & \\
		\midrule
		model type 1 & 200 & 50 & 5 & 500 & 5,000 & 10\%\\
		model type 2 & 50 & 100 & 5 & 1,000 & 20,000 & 5\% \\
		model type 3 & 100 & 200  & 10 & 4,000 & 80,000 & 5\% \\
		\bottomrule
	\end{tabular}%
	\label{tab:model_specs}%
\end{table}%
For $\phi_1$, we randomly select $s$ columns to be nonzero and fill them with uniform random variables on $[-0.5,0.5]$. The same procedure is used to generate $\phi_2$. To guarantee stationarity of the simulated process, we need to impose some restriction on the maximum eigenvalue of the coefficient matrices. Rewrite the VAR(2) in \cref{eq:var2_simu} in a VAR(1) format (see equation 10.1.11 in \cite{hamilton1994time})
\begin{equation} \label{eq:var1_format}
	\mathcal{Y}_t = \mathbf{F} \mathcal{Y}_{t-1} + \mathbf{v},
\end{equation}
where
\begin{equation} \label{eq:F_mat}
	\mathbf{F} =
	\begin{bmatrix}
	\phi_1 & \phi_2\\
	\mathbf{I}_d & \mathbf{0} \\
	\end{bmatrix}.
\end{equation}
Given $\phi_1$ and $\phi_2$, we recursively shrink their value until the maximum eigenvalue of $\mathbf{F}$ is less than $1$, and this will make $y_t$ in  \cref{eq:var2_simu} covariance-stationary.

(iii) Create a Toeplitz matrix $\tilde{\Omega}$, whose $ij$th element is  $\rho^{|i-j|}$ with $\rho$ chosen to be $0.5$. For a given signal-to-noise (SNR) level, choose a value for $\sigma^2$ so that it meets the desired SNR level, i.e., $\text{SNR} = \frac{\lambda_{\text{max}}(\mathbf{F}) }{\sigma^2\lambda_{\text{max}}(\tilde{\Omega})}$. Let the error variance-covariance matrix be $\Omega = \sigma^2 \tilde{\Omega}$. Draw the i.i.d. $d \times 1$ error vector $u_t$ from $N(\mathbf{0}, \Omega)$. We consider three SNR levels in our study: $0.5$, $1.0$, and $3.0$.

(iv) With $\phi_1$, $\phi_2$, and $\Omega$, we can simulate the process according to \cref{eq:var2_simu}.

(v) Run various methods (discussed below), including LS-Boost1 and LS-Boost2, on the simulated data. For each method, we estimate the model based on a training data set, choose the tuning parameter on a validation set, and evaluate the model performance on a test set. The number of observations in both validation and test sets is $200$. All model metrics reported are averages of $100$ replications. Alternatively, one can use an information criterion to determine the tuning parameter. Since there are quite a few information criteria to choose from, using a validation set provides a simple and uniform approach for model selection across different methods.

(vi) For each method, we report the following metrics: Mean squared error (MSE) for estimated coefficients: $\text{MSE} = \lVert \hat{\bm{\phi}} - \bm{\phi} \rVert_{2}^2 /(p\times d^2)$; Mean squared prediction error (MSPE) on the test set: $\text{MSPE} = \lVert \mathbf{X}(\hat{\bm{\phi}} - \bm{\phi}) \rVert_{2}^2 /(T_{\text{test}}\times d)$, where $T_{\text{test}} = 200$; false positive rate (FPR): $\text{FPR} = \text{FP}/(\text{TN}+\text{FP})$ and false negative rate (FNR): $\text{FNR} = \text{FN}/(\text{TP}+\text{FN})$, where TN, FP, TP, and FN are number of true zeros, false nonzeros, true nonzeros, and false nonzeros; F-score: $\text{F-score} = 2\text{TP}/(2\text{TP} + \text{FP} + \text{FN})$. The F-score measures the accuracy of support recovery, and it balances FPR and FNR. Finally, we also report the model size by counting the number of nonzeros in the selected model from validation for each method. All evaluation metrics are averages over $100$ replications.

\subsubsection{Methods under consideration}

We mentioned in \Cref{sec:intro} that there is now a long list of statistical methods available for estimation in high-dimensional VARs. Our goal is not an extensive comparison of different methods but the study of estimation and use of the proposed \textit{p}-value in LS-Boost. Hence, we select only a few of the other methods for demonstration purposes.

The first method to consider is the lasso. Various papers have used the lasso in high-dimensional VARs, and the actual implementation of this technique varies; variants such as the adaptive lasso or the group lasso can also be used. One common approach is to apply the lasso to each column of $\mathbf{Y}$ in \cref{eq:varp_matrix} separately. However, in a VAR setup, it is important to do estimation jointly to incorporate possible correlation across the $d$ variables. We instead use the multi-response approach in the R package \texttt{glmnet} by setting the option $\texttt{family = "mgaussian"}$. For the $j$th column in $\mathbf{X}$ in \cref{eq:varp_matrix}, the algorithm uses a group lasso penalty to penalize the entire $j$th row of $\bm{\phi}$. Hence, the modeling strategy is: a column of $\mathbf{X}$ will be included in the model for all $d$ response variables or excluded for all $d$ response variables. This is exactly the same modeling strategy LS-Boost2 adopts. We note that a row in $\bm{\phi}$ is a column in either $\phi_1$ or $\phi_2$. Given that we simulate $\phi_1$ and $\phi_2$ by randomly assign $s$ nonzero columns, this way of applying the lasso to the simulated data will give the lasso a slight edge over LS-Boost1.

We also consider two other recent methods that are suitable for estimation in high-dimensional VARs: the sparse orthogonal factor regression (SOFAR) in \cite{uematsu2019sofar} and the sparse reduced-rank regression (SRRR) in \cite{chen2012srrr}, both of which are implemented in the R package \texttt{rrpack}. Both method requires a prespecified rank parameter. We tune this \texttt{nrank} parameter on the grid of $\{1,2,3,4,5\}$. These two methods also uses the lasso to penalize certain aspects of the coefficient matrix $\bm{\phi}$ to generate sparsity, and we use $100$ for the length of the lasso tuning sequence. All other model control parameters are the default values provided by the \texttt{rrpack} package. Both methods are very sophisticated and flexible in application. For example, the \texttt{sofar} function in the package has more then $10$ model control parameters. It is almost impossible to experiment with all model parameter configurations. We use the default values for simplicity purposes.  

Since we use $5$ rank parameters and a sequence of $100$ penalty parameters in both the SOFAR and the SRRR, we set the length of tuning sequence in the lasso and LS-Boost to $500$. In both LS-Boost1 and LS-Boost2, we use the learning rate $0.1$, and train the model up to $500$ iterations and send these solutions to a validation set for model selection and a test set for model evaluation.

Finally, after getting the boosting solution paths and the \textit{p}-values, we use $5\%$ as the \textit{p}-value cutoff to remove nonzeros in the boosting estimates and call these estimates LS-Boost1p and LS-Boost2p, respectively. Apply LS-Boost1p and LS-Boost2p to the validation set for model selection and to the test set for model evaluation. As it is discussed in \Cref{sec:discussions}, this is a valid way to use the proposed \textit{p}-value.

\subsubsection{The simulation results and the evaluation metrics}

\Cref{tab:modeltype1,tab:modeltype2,tab:modeltype3} report the simulation results for the three models described in \Cref{tab:model_specs} at three SNR levels. Our major conclusions are the proposed \textit{p}-value can help LS-Boost improve FPR and F score, and it can also help LS-Boost significantly reduce model size without compromising MSE or MSPE too much.

\textit{FPR, FNR and F score}. From these tables, it is clear that all regularized methods under consideration in high dimensional VAR can have very high FPR, implying that their Type I error rate would be too high to make these model useful for interpretation. After applying the \textit{p}-value to remove many nonzeros, the FPR for LS-Boost1p and LS-Boost2p is roughly around $5\%$, bringing the FPR under control. A downside is the FNR increases as the FPR decreases. Inspecting the F score will be helpful. Overall, the \textit{p}-value-adjusted models have a higher F score compared to that of LS-Boost1 and LS-Boost2 and the other two methods.

\textit{Model size, MSE and MSPE}. The MSE and MSPE of all methods are more or less on the same scale. When the model dimension is very large in \Cref{tab:modeltype3}, the lasso and LS-Boost start to outperform the SOFAR and SRRR by a small margin. Notably, LS-Boost1p and LS-Boost2p usually deliver similar MSE and MSPE with far fewer parameters, as can be seen from the Model size column in all tables. If we compare the MSPE of LS-Boost1 and LS-Boost1p, LS-Boost1p can improve MSPE in \Cref{tab:modeltype1} and its performance in \Cref{tab:modeltype3} is a little worse. Overall, based on these tables, we cannot conclude that applying \textit{p}-values to the original LS-Boost models will always improve MSE or MSPE. However, it should also be noted that estimators such as LS-Boost1p and LS-Boost2p are lot more \textit{effective} in using nonzero parameters. Consider the MSPE in \Cref{tab:modeltype3} when $\text{SNR}=3.0$, LS-Boost2 uses $12,832$ parameters to deliver an MSPE of $0.0823$ and the lasso uses $28,050$ parameters to obtain an MSPE of $0.0813$. But it takes only $4,501$ parameters for LS-Boost2p to yield a similar MSPE of $0.0862$. Similar pattern can be observed throughout all tables and also in our application example.

\begin{table}[htp] \centering
	\begin{center}
		\caption{Evaluation metrics for  model type $1$ with $T=200, d = 50, s=5$} 
		\label{tab:modeltype1} 
		\begin{threeparttable}
			\begin{tabular}{lrrrrrr} 
				\toprule
				Estimator& MSE & MSPE & FPR& FNR& F score &Model size\\\hline
				&			\multicolumn{6}{c}{$\text{SNR} = 0.5$} \\
				\cmidrule(lr){2-7}	
				SOFAR & 0.0039 & 0.3907 & 0.262 & 0.380 & 0.320 & 1487.7 \\
				SRRR  & 0.0043 & 0.4064 & 0.985 & 0.004 & 0.183 & 4929.5 \\
				Lasso & 0.0040 & 0.3939 & 0.423 & 0.274 & 0.264 & 2266.0 \\
				LS-Boost1 & 0.0043 & 0.4167 & 0.586 & 0.000 & 0.282 & 3135.0 \\
				LS-Boost2 & 0.0041 & 0.3975 & 0.321 & 0.324 & 0.305 & 1780.5 \\
				LS-Boost1p & 0.0042 & 0.4081 & 0.064 & 0.575 & 0.424 &  501.9 \\
				LS-Boost2p & 0.0041 & 0.3981 & 0.049 & 0.585 & 0.448 &  427.3 \\
				
				&			\multicolumn{6}{c}{$\text{SNR} = 1.0$} \\
				\cmidrule(lr){2-7}
				SOFAR & 0.0039 & 0.1953 & 0.259 & 0.382 & 0.320 & 1476.2 \\
				SRRR  & 0.0040 & 0.1958 & 0.849 & 0.075 & 0.194 & 4284.0 \\
				Lasso & 0.0040 & 0.1970 & 0.423 & 0.274 & 0.264 & 2266.0 \\
				LS-Boost1 & 0.0043 & 0.2084 & 0.586 & 0.000 & 0.282 & 3135.0 \\
				LS-Boost2 & 0.0041 & 0.1988 & 0.321 & 0.324 & 0.305 & 1780.5 \\
				LS-Boost1p & 0.0042 & 0.2041 & 0.064 & 0.575 & 0.424 &  501.9 \\
				LS-Boost2p & 0.0041 & 0.1990 & 0.049 & 0.585 & 0.448 &  427.3 \\
				
				&			\multicolumn{6}{c}{$\text{SNR} = 3.0$} \\
				\cmidrule(lr){2-7}	
				SOFAR & 0.0039 & 0.0651 & 0.243 & 0.391 & 0.332 & 1396.9 \\
				SRRR  & 0.0038 & 0.0637 & 0.225 & 0.404 & 0.335 & 1309.0 \\
				Lasso & 0.0040 & 0.0657 & 0.423 & 0.282 & 0.261 & 2263.0 \\
				LS-Boost1 & 0.0043 & 0.0695 & 0.584 & 0.000 & 0.282 & 3126.0 \\
				LS-Boost2 & 0.0041 & 0.0663 & 0.315 & 0.333 & 0.304 & 1753.0 \\
				LS-Boost1p & 0.0042 & 0.0680 & 0.065 & 0.572 & 0.426 &  505.3 \\
				LS-Boost2p & 0.0041 & 0.0664 & 0.050 & 0.582 & 0.448 &  434.1 \\
					
				\bottomrule 
			\end{tabular}
			\begin{tablenotes}[flushleft]
				\item[] \textit{Notes}: All evaluation metrics reported in the table are averages over $100$ replications. In model type 1, the true nonzeros are $500$ and the total number of parameters is $5,000$. See \Cref{sec:simulation_setup} for the computation of MSE, MSPE, FPR, FNR, and F score. Model size is the total number of nonzeros in the selected model for each method.
			\end{tablenotes}
		\end{threeparttable}
	\end{center} 
\end{table}

\begin{table}[htp] \centering
	\begin{center}
		\caption{Evaluation metrics for  model type $2$ with $T=50, d = 100, s=5$} 
		\label{tab:modeltype2} 
		\begin{threeparttable}
			\begin{tabular}{lrrrrrr} 
				\toprule
				Estimator& MSE & MSPE & FPR& FNR& F score&Model size\\\hline
				&			\multicolumn{6}{c}{$\text{SNR} = 0.5$} \\
				\cmidrule(lr){2-7}	
				
				SOFAR & 0.0028 & 0.8827 & 0.315 & 0.372 & 0.169 &  6609.9 \\
				SRRR  & 0.0025 & 0.8428 & 0.495 & 0.250 & 0.135 & 10156.0 \\
				Lasso & 0.0026 & 0.8553 & 0.287 & 0.398 & 0.173 &  6046.0 \\
				LS-Boost1 & 0.0031 & 0.9559 & 0.313 & 0.000 & 0.260 &  6952.0 \\
				LS-Boost2 & 0.0028 & 0.8802 & 0.183 & 0.446 & 0.230 &  4035.0 \\
				LS-Boost1p & 0.0031 & 0.9523 & 0.059 & 0.621 & 0.304 &  1502.1 \\
				LS-Boost2p & 0.0029 & 0.9000 & 0.040 & 0.666 & 0.320 &  1097.0 \\
				
				&			\multicolumn{6}{c}{$\text{SNR} = 1.0$} \\
				\cmidrule(lr){2-7}	
				SOFAR & 0.0028 & 0.4414 & 0.323 & 0.364 & 0.168 & 6777.4 \\
				SRRR  & 0.0025 & 0.4145 & 0.426 & 0.280 & 0.147 & 8812.0 \\
				Lasso & 0.0026 & 0.4277 & 0.286 & 0.398 & 0.173 & 6043.0 \\
				LS-Boost1 & 0.0031 & 0.4780 & 0.313 & 0.000 & 0.260 & 6952.0 \\
				LS-Boost2 & 0.0028 & 0.4401 & 0.183 & 0.446 & 0.230 & 4035.0 \\
				LS-Boost1p & 0.0031 & 0.4761 & 0.059 & 0.621 & 0.304 & 1502.1 \\
				LS-Boost2p & 0.0029 & 0.4500 & 0.040 & 0.666 & 0.320 & 1097.0 \\
				
				&			\multicolumn{6}{c}{$\text{SNR} = 3.0$} \\
				\cmidrule(lr){2-7}
				SOFAR & 0.0029 & 0.1466 & 0.334 & 0.364 & 0.163 & 6981.1 \\
				SRRR  & 0.0023 & 0.1329 & 0.195 & 0.387 & 0.235 & 4313.0 \\
				Lasso & 0.0026 & 0.1415 & 0.290 & 0.411 & 0.168 & 6106.0 \\
				LS-Boost1 & 0.0031 & 0.1578 & 0.332 & 0.000 & 0.250 & 7300.0 \\
				LS-Boost2 & 0.0028 & 0.1458 & 0.182 & 0.450 & 0.228 & 4009.0 \\
				LS-Boost1p & 0.0031 & 0.1579 & 0.058 & 0.623 & 0.305 & 1478.1 \\
				LS-Boost2p & 0.0029 & 0.1490 & 0.040 & 0.667 & 0.320 & 1092.4 \\
							
				\bottomrule 
			\end{tabular}
			\begin{tablenotes}[flushleft]
				\item[] \textit{Notes}: All evaluation metrics reported in the table are averages over $100$ replications. In model type 2, the true nonzeros are $1,000$ and the total number of parameters is $20,000$. See \Cref{sec:simulation_setup} for the computation of MSE, MSPE, FPR, FNR, and F score. Model size is the total number of nonzeros in the selected model for each method.
			\end{tablenotes}
		\end{threeparttable}
	\end{center} 
\end{table}

\begin{table}[htp] \centering
	\begin{center}
		\caption{Evaluation metrics for  model type $3$ with $T=100, d = 200, s=10$} 
		\label{tab:modeltype3} 
		\begin{threeparttable}
			\begin{tabular}{lrrrrrr} 
				\toprule
				Estimator& MSE & MSPE & FPR& FNR& F score&Model size\\\hline
				&			\multicolumn{6}{c}{$\text{SNR} = 0.5$} \\
				\cmidrule(lr){2-7}	
				SOFAR & 0.0015 & 0.5919 & 0.242 & 0.423 & 0.190 & 20678.7 \\
				SRRR  & 0.0015 & 0.5865 & 0.546 & 0.224 & 0.128 & 44628.0 \\
				Lasso & 0.0012 & 0.4889 & 0.331 & 0.349 & 0.164 & 27786.0 \\
				LS-Boost1 & 0.0013 & 0.5476 & 0.254 & 0.000 & 0.307 & 23292.0 \\
				LS-Boost2 & 0.0012 & 0.4951 & 0.136 & 0.446 & 0.279 & 12590.0 \\
				LS-Boost1p & 0.0014 & 0.5532 & 0.063 & 0.668 & 0.263 &  6137.0 \\
				LS-Boost2p & 0.0013 & 0.5181 & 0.042 & 0.701 & 0.286 &  4404.7 \\
				
				&			\multicolumn{6}{c}{$\text{SNR} = 1.0$} \\
				\cmidrule(lr){2-7}	
				SOFAR & 0.0015 & 0.2960 & 0.243 & 0.422 & 0.189 & 20778.0 \\
				SRRR  & 0.0014 & 0.2911 & 0.503 & 0.244 & 0.134 & 41238.0 \\
				Lasso & 0.0012 & 0.2444 & 0.331 & 0.349 & 0.164 & 27786.0 \\
				LS-Boost1 & 0.0013 & 0.2738 & 0.254 & 0.000 & 0.307 & 23292.0 \\
				LS-Boost2 & 0.0012 & 0.2476 & 0.136 & 0.446 & 0.279 & 12590.0 \\
				LS-Boost1p & 0.0014 & 0.2766 & 0.063 & 0.668 & 0.263 &  6137.0 \\
				LS-Boost2p & 0.0013 & 0.2591 & 0.042 & 0.701 & 0.286 &  4404.7 \\
				
					&			\multicolumn{6}{c}{$\text{SNR} = 3.0$} \\
				\cmidrule(lr){2-7}	
				SOFAR & 0.0015 & 0.0984 & 0.256 & 0.420 & 0.183 & 21777.4 \\
				SRRR  & 0.0014 & 0.0950 & 0.334 & 0.340 & 0.165 & 28030.0 \\
				Lasso & 0.0012 & 0.0813 & 0.335 & 0.344 & 0.164 & 28050.0 \\
				LS-Boost1 & 0.0013 & 0.0911 & 0.254 & 0.000 & 0.310 & 23340.0 \\
				LS-Boost2 & 0.0012 & 0.0823 & 0.139 & 0.440 & 0.278 & 12832.0 \\
				LS-Boost1p & 0.0014 & 0.0919 & 0.064 & 0.668 & 0.263 &  6159.0 \\
				LS-Boost2p & 0.0013 & 0.0862 & 0.043 & 0.700 & 0.284 &  4501.2 \\
				
				\bottomrule 
			\end{tabular}
			\begin{tablenotes}[flushleft]
				\item[] \textit{Notes}: All evaluation metrics reported in the table are averages over $100$ replications. In model type 3, the true nonzeros are $4,000$ and the total number of parameters is $80,000$. See \Cref{sec:simulation_setup} for the computation of MSE, MSPE, FPR, FNR, and F score. Model size is the total number of nonzeros in the selected model for each method.
			\end{tablenotes}
		\end{threeparttable}
	\end{center} 
\end{table}

\section{Application} \label{sec:application}

In this section, we apply the proposed \textit{p}-value to the VAR modeling of the monthly macroeconomic data set provided in \cite{mccracken2016data}. This dataset is frequently updated and well-maintained, and we use version 2022-02 that records $128$ economic variables from $1959/01$ to $2022/01$. After removing missing observations and performing the suggested transformation such as difference and twice difference as described by the \texttt{TCODE} variable in the paper, we end up with $104$ variables and $758$ observations. Out of these $758$ observations, $50\%$, $25\%$ and $25\%$ are used for training, validation and test, respectively. An intercept is included in all model fitting.

Exactly the same tuning process used in the simulation section is applied to each method in this application except that for SOFAR and SRRR. We increase the penalty parameter sequence for these two methods from $100$ to $500$ while still tuning over five rank values, giving these two methods a large grid for the search of a good solution. SOFAR sometimes does not identify any model on the validation set and we will use the model selected by AIC instead. Both SOFAR and SRRR have many model parameters to tune in the R package \texttt{rrpack}. We will not make further effort to tune these parameters in their R code while keeping in mind that their predication performance in our reported table can be likely improved if more tuning is implemented. We focus on the proposed \textit{p}-value and its impact on the LS-Boost procedure.

\begin{table}[htp] \centering
	\begin{center}
		\caption{MSPE and model size in application} 
		\label{tab:application_MSPE} 
		\begin{threeparttable}
			\begin{tabular}{rrrrrrcc}  
				\toprule
				& SOFAR & SRRR & Lasso &LS-Boost1& LS-Boost2 & LS-Boost1p &  LS-Boost2p \\
				\midrule
				
				&			\multicolumn{7}{c}{MSPE} \\
				\cmidrule(lr){2-8}
				
				VAR(1) & 3.399 & 0.176 & 270.718 & 391.494 & 391.494 & 0.157 & \textbf{0.157} \\
				VAR(2) & 6.750  & 0.190  & 182.781 & 255.170 & 354.143 & 0.159 & \textbf{0.157} \\
				VAR(3) & 3.247 & 0.211 & 105.392 & 230.217 & 200.226 & 0.168 & \textbf{0.158} \\
				VAR(4) & 3.788 & 0.653 & 87.728 & 187.724 & 213.099 & 0.171 & \textbf{0.159} \\
				&			\multicolumn{7}{c}{Model size (No. of nonzero parameter estimates)} \\
				\cmidrule(lr){2-8}
				VAR(1) & 528   & 1144  & 6240  & 7072  & 7072  & 411   & \textbf{411} \\
				VAR(2) & 968   & 1352  & 8216  & 14352 & 10296 & 425   & \textbf{410} \\
				VAR(3) & 32448 & 1560  & 6656  & 15288 & 7488  & 519   & \textbf{411} \\
				VAR(4) & 43264 & 2184  & 7904  & 20384 & 10608 & 618   & \textbf{410} \\
				
				\bottomrule
			\end{tabular}
			\begin{tablenotes}[flushleft]
				\setlength\labelsep{0pt}
				\item[] \textit{Notes}: The sample size for the training, validation and test sets are $378,189$ and $189$. The MSPE is computed for the test set. Intercept is excluded when computing the model size of each method. The total number of parameters, excluding the intercept, for VAR(1), VAR(2), VAR(3), and VAR(4) are $10816, 21632, 32448$, and $43264$, respectively.
			\end{tablenotes}
		\end{threeparttable}
	\end{center} 
\end{table}

\Cref{tab:application_MSPE} reports the main results of the application. Two model evaluation metrics are included: MSPE and model size. The MSPE is computed based on the test set and model size is the number of nonzero parameter estimates for the model selected from the validation data set for each method. The MSPE shows that both Lasso and LS-Boost fail to produce a result comparable to that of  SOFAR or SRRR. This should not be construed as a failure of the lasso method in general. Our earlier discussion points out that there is no unique way to implement the lasso in VAR. Other forms of the lasso can also be implemented. We discuss a few aspects of \Cref{tab:application_MSPE} in the following.

\subsection{Model sparsity} 

One reason that LS-Boost fails to deliver any meaningful result is probably due to the noise in the data. Noise and correlation make it very difficult for the boosting algorithm to focus on the most essential subset of variables; instead, the algorithm keeps adding new variables to the model. The model size quickly becomes out of control with the number of nonzeros running well into thousands or tens of thousands. For example, in a VAR(4) model, LS-Boost1 gives a model with $2,0384$ nonzero parameter estimates. Given the VAR system have $104$ equations, it implies, on average, each equation has $196$ nonzero parameters on the r.h.s. The possible high FPR (and possibly high FDR) translates to a very poor performance on the test set. Using the \textit{p}-value salvages the boosting method. By simply removing a coefficient estimate with a \textit{p}-value larger than $5\%$, we significantly reduces model size. To appreciate the reduction in model size, let us consider the VAR(4) model. A VAR(4) has a total of $43,264 (=104\times 104 \times 4)$ parameters, excluding the intercept. The algorithm LS-Boost2p gives only $410$ nonzeros, achieving a sparsity of about $0.9\%$, reducing the average number of variables per equation from $196$ to about $4$, making the selected model highly interpretable.

\subsection{Model stability} 

A surprising benefit is \textit{p}-values can help control model stability. Use the LS-Boost2p as an example. As the model changes from VAR(1) to VAR(4), the model size barely changes, indicating that the $5\%$ critical value removes most of the nonzeros. In fact, for LS-Boost2p, all the nonzeros appear in the coefficient matrix for $\mathbf{Y}_{t-1}$ and all coefficients in the second, third, and fourth lags are zero. \Cref{fig:circosNEG,fig:circosPOS} give an illustration of the VAR(1) coefficients. Variables appear in \Cref{fig:circosNEG,fig:circosPOS} are described in \Cref{tab:var_names}.

\begin{table}[htp] \centering
	\begin{center}
		\caption{Variable definitions in \Cref{fig:circosNEG,fig:circosPOS,fig:circosPOSlarge}} 
		\label{tab:var_names} 
		\begin{threeparttable}
			\begin{tabular}{lcl} 
				\toprule
				\text{Variable}& Group No.  & Description  \\
				\midrule
				\texttt{AAAFFM} & 6 I.R.E.R. & Moody's Aaa corporate bond minus federal funds rate \\
				\texttt{BAAFFM} & 6 I.R.E.R. & Moody's Baa corporate bond minus federal funds rate \\
				\texttt{GS1} & 6 I.R.E.R.  & 1-year treasury rate  \\
				\texttt{T10YFMM} & 6 I.R.E.R.  & 10-year treasury constant maturity minus federal funds rate \\
				\texttt{T1YFMM} & 6  I.R.E.R.  & 1-year treasury constant maturity minus federal funds rate  \\
				\texttt{T5YFMM} & 6  I.R.E.R.  & 5-year treasury constant maturity minus federal funds rate  \\
				\texttt{TB3SMFMM} & 6  I.R.E.R.  & 3-month treasury constant maturity minus federal funds rate  \\
				\texttt{TB6SMFMM} & 6  I.R.E.R.  & 6-month treasury constant maturity minus federal funds rate  \\
				\bottomrule				 
			\end{tabular}
			\begin{tablenotes}[flushleft]
				\item[] \textit{Notes}: All variables are from Group $6$ described in \cite{mccracken2016data}. I.R.E.R. is the abbreviation for the group name Interest Rate and Exchange Rates. 
			\end{tablenotes}
		\end{threeparttable}
	\end{center} 
\end{table}

\begin{figure}[htp]
	\centering
	\subfloat[Negative links between $\mathbf{Y}_t$ and $\mathbf{Y}_{t-1}$ 
	\label{fig:circosNEG}]{{\includegraphics[width=0.5\linewidth,scale=1.0]{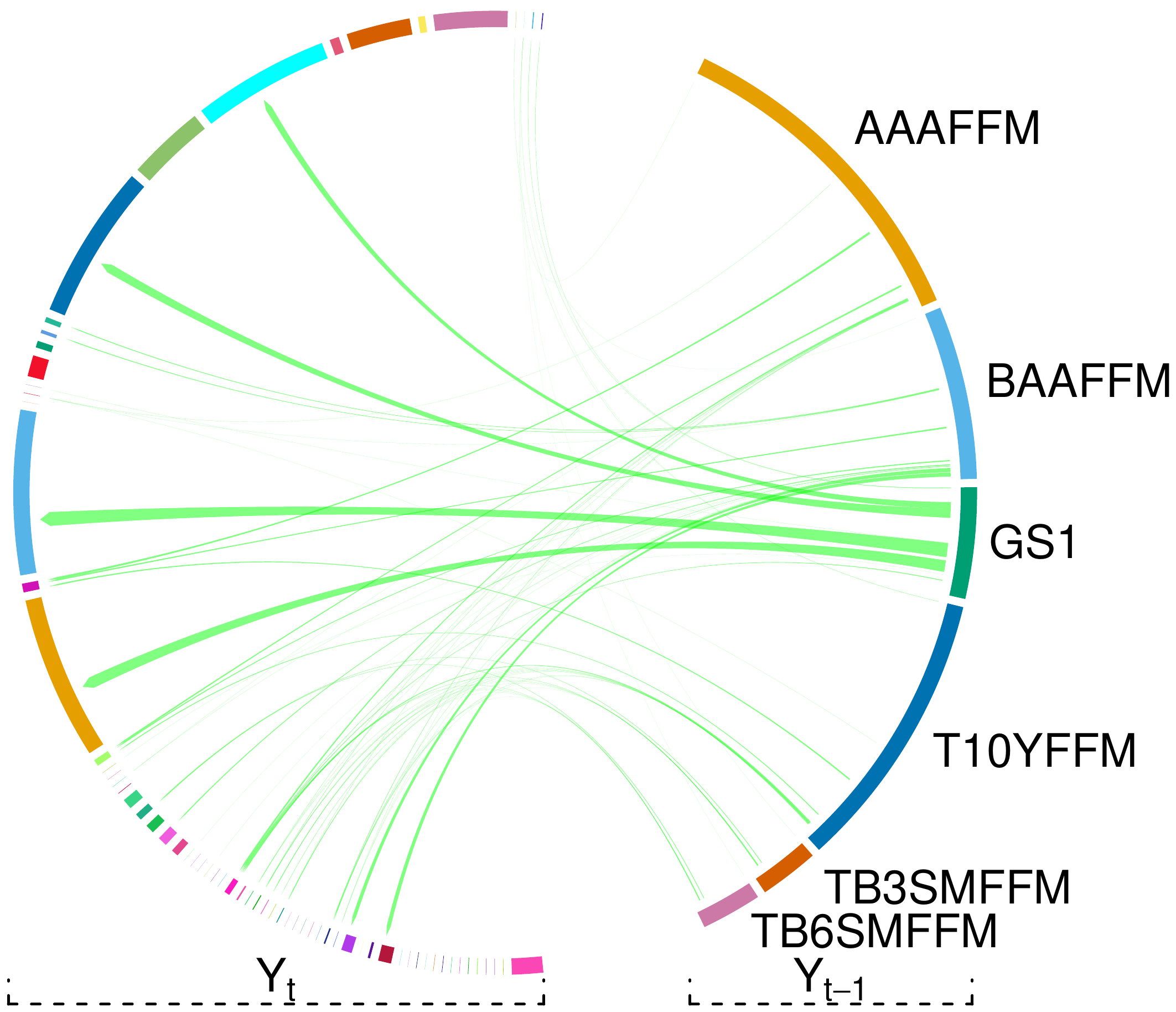} }}%
	\subfloat[Positive links between $\mathbf{Y}_t$ and $\mathbf{Y}_{t-1}$ 
	\label{fig:circosPOS}]{{\includegraphics[width=0.5\linewidth,scale=1.0]{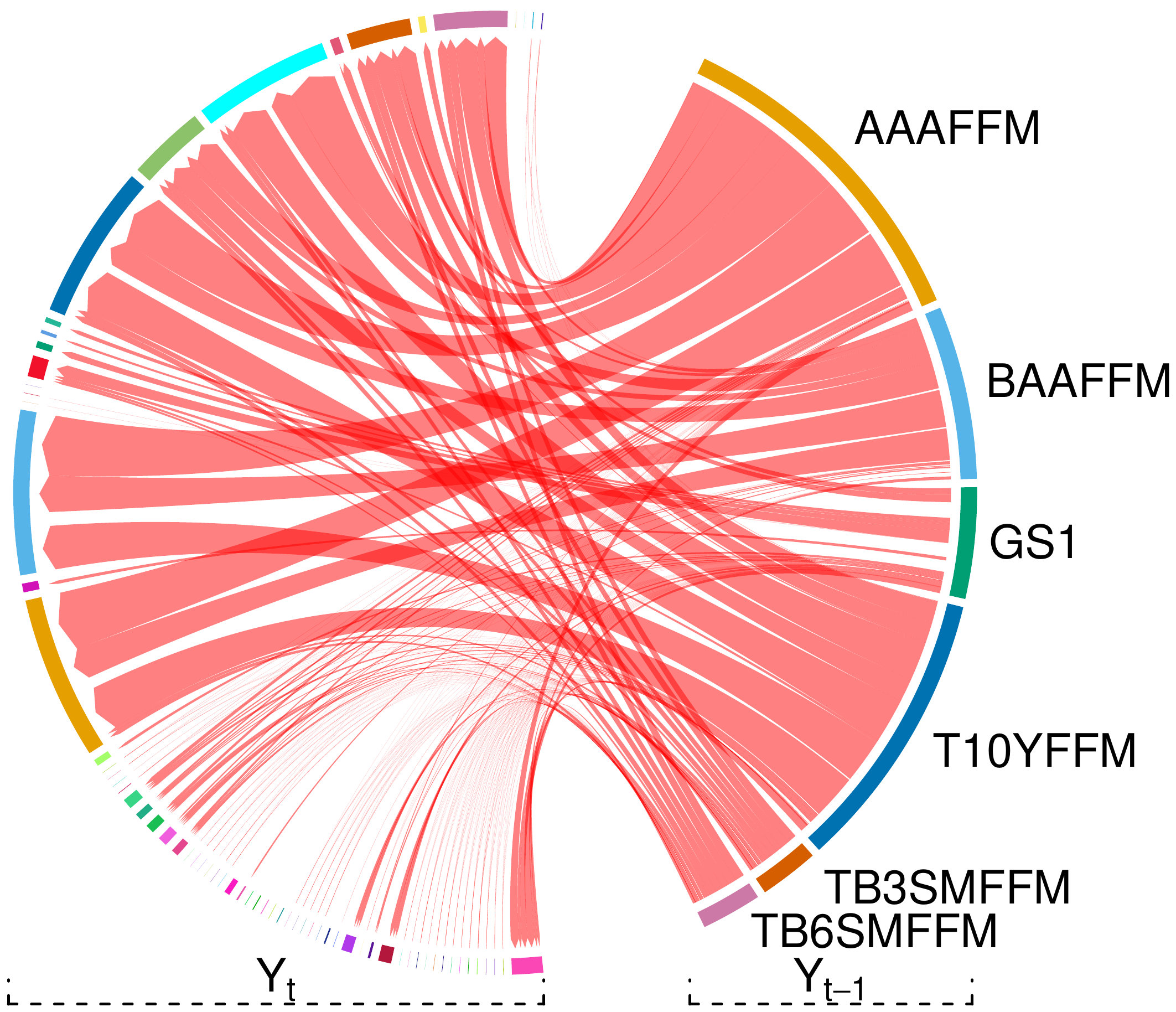} }}%
	\caption{\Cref{fig:circosNEG,fig:circosPOS} describe the negative and positive links in the VAR(1) coefficient matrix from LS-Boost2p, respectively. In each subfigure, the right half circle includes variables in $\mathbf{Y}_{t-1}$ with a nonzero (negative or positive) coefficient on $\mathbf{Y}_{t}$; the left half circle includes variables in $\mathbf{Y}_t$ that the $6$ variables affect. Green and red links indicate negative and positive coefficient estimates, respectively. There are a total of $411$ links, equivalent to $411$ nonzeros reported in \Cref{tab:application_MSPE} for LS-Boost2p.}%
\end{figure}

Out of the $104$ variables, $6$ of them are consistently selected by the boosting procedure. It is important to recall that, for LS-Boost2p, all nonzeros concentrate in the coefficient for $\mathbf{Y}_{t-1}$. Hence, these $6$ variables are the only (statistically significant) source that drives $\mathbf{Y}_t$, not only in VAR(1) but also in higher-order VARs (at least up to oder $4$). Consider \Cref{fig:circosNEG} that plots only the negative coefficients in the VAR(1) model based on LS-Boost2p. The green color for the links indicates a negative coefficient and all links start from the right half circle to the left half circle with an arrow pointing to the left, mimicking the usual setup in a regression where $\mathbf{Y}_{t-1}$ appears on the r.h.s. and $\mathbf{Y}_t$ on the l.h.s. The width of a link is proportional to the absolute value of the corresponding coefficient estimate, and a wider link indicates a negative coefficient with a larger absolute value. The colored sectors have different width. A wider sector on the right half circle suggests a variable has a larger overall influence (sum of absolute values of the coefficients). For example, \texttt{AAAFFM} has the widest sector due to its overall large coefficients. The sector for \texttt{AAAFFM} is only partially filled in \Cref{fig:circosNEG} because most of the corresponding coefficients are positive; the sector will be completely filled if we combine the links in \Cref{fig:circosNEG,fig:circosPOS}.

The links reveals various properties in the data. Consider the \texttt{GS1} sector in \Cref{fig:circosNEG} as an example. \texttt{GS1} is the 1-year treasury rate, a key variable for short-term interest rate. There are four relatively wide arrows leaving the sector \texttt{GS1} and pointing at four variables in $\mathbf{Y}_t$. These four negative links may puzzle a researcher at a first look, but they are largely due to the way variables are defined in \cite{mccracken2016data}. For example, we see \texttt{GS1} has a relatively large, negative impact on \texttt{AAAFFM}, the corporate bond yield spread. We would expect that bond yield will increase as the short-term interest rate increases. A negative link is counter-intuitive. Notice that \texttt{AAAFFM} is defined as corporate bond yield minus the federal funds rate (\texttt{FEDFUNDS}). Using the $758$ observations in the data, we add \texttt{FEDFUNDS} back to \texttt{AAAFFM} and compute its correlation with \texttt{GS1}, and the result is $0.899$, indicating these two time series are positively correlated. Hence, it is important to interpret the result between \texttt{GS1} and \texttt{AAAFFM} in \Cref{fig:circosNEG} as the negative relationship between the short-term interest rate and corporate bond yield \textit{spread}. Similar explanations applies to other links between \texttt{GS1} and other bond yields or long-term interest rates in that figure.

\Cref{fig:circosPOS} describes all positive links between the $6$ variables and $\mathbf{Y}_t$. We observe that most of the $411$ coefficient estimates from LS-Boost2p are positive. The $6$ variables have relatively large autoregressive coefficients, but there are also many positive links to other variables. In a VAR(1) model with $104$ variables, a dense coefficient matrix will have $10,816 (=104 \times 104)$ links. Hence \Cref{fig:circosPOS} is already a very sparse representation of the coefficient matrix. To further help interpret the result, we remove all coefficients that are $< 0.1$ and re-draw the links in \Cref{fig:circosPOSlarge}.

\begin{figure}[htp]
	\centering
	\includegraphics[width=0.8\linewidth]{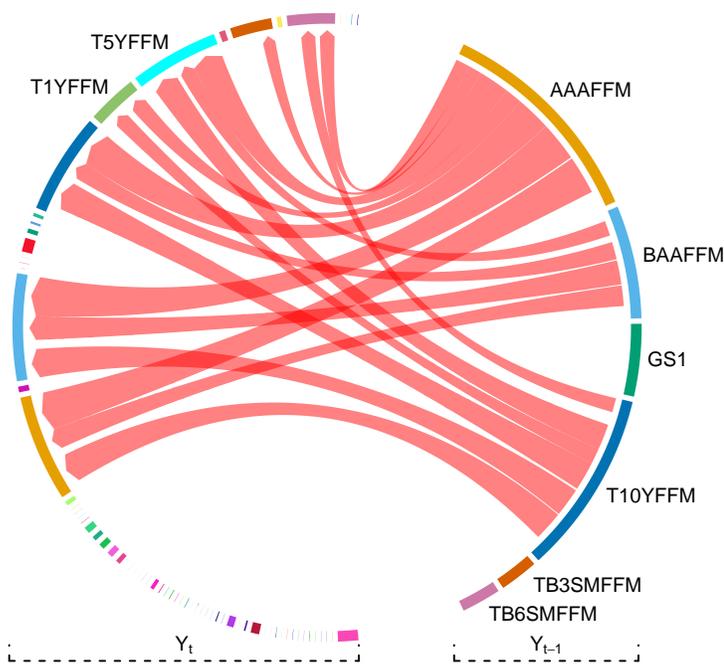}
	\caption{Large positive links with coefficient estimates $\geq 0.1$ in \Cref{fig:circosPOS}.}
	\label{fig:circosPOSlarge}
\end{figure}

\Cref{fig:circosPOSlarge} is now even more interpretable. It clearly reveals that a risk premium such as \texttt{AAAFFM}, \texttt{BAAFFM}, and \texttt{T10YFFM} has a highly persistent pattern with its own lag, and their impact on other variables is also obvious. The dynamic feedback mechanism is asymmetric in a VAR. For example, \texttt{T10YFFM} has a positive effect on \texttt{T1YFFM} but not vice versa. Although we selectively choose positive coefficients $\geq 0.1$, one should also take into consideration the scale of the data when measuring the impact of a lag variable.

\Cref{fig:circosNEG,fig:circosPOS,fig:circosPOSlarge} illustrate many important economic relationships uncovered from high-dimensional VAR analysis. The \textit{p}-value approach significantly shrinks the model size and keep the model stable in high dimension. It provides a valuable tool for empirical researcher to investigate large data sets.

\subsection{Sequential testing}


\Cref{sec:discussions} points out the sequential nature of our testing procedure. We further illustrate it with a full plot of the $500$ \textit{p}-values for each of the $6$ variables shown in \Cref{fig:circosNEG,fig:circosPOS,fig:circosPOSlarge}.  Our empirical example has $104$ equations. We select the equation for civilian unemployment rate (\texttt{UNRATE}) as an example.

\begin{figure}[t]
	\centering
	\includegraphics[width=0.8\linewidth]{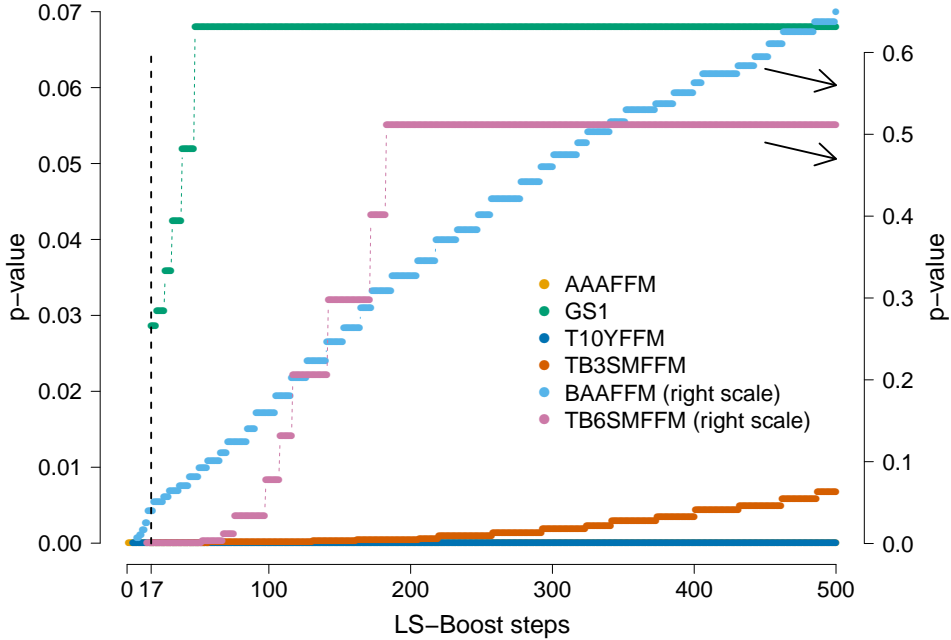}
	\caption{\Cref{fig:plot_pval} plots the entire path of the \textit{p}-values for the $6$ variables in the equation for Civilians Unemployment Rate (\texttt{UNRATE}). Use the right scale to read the \textit{p}-values for $\texttt{BAAFFM}$ and \texttt{TB6SMFFM}. The validation step chooses to stop at step $17$, marked by a vertical dotted line.}
	\label{fig:plot_pval}
\end{figure}

\Cref{fig:plot_pval} plots the path of \textit{p}-values for all $6$ variables. The validation step selects the model at step $17$. Three variables, \texttt{AAAFFM}, \texttt{T10YFFM} and \texttt{TB3SMFFM} have very small \textit{p}-values throughout the $500$ boosting steps. The \textit{p}-values for \texttt{GS1} is also small, but climbed above the $5\%$ cutoff after around step $50$, but is below $5\%$ at step $17$. The \textit{p}-values for \texttt{BAAFFM} and \texttt{TB6SMFFM}, however, become very large in later steps and can exceed $50\%$. Had the validation procedure stopped at, say, step $300$, both variables would have been excluded from the model.

No Bonferroni correction is used at each step when excluding variables. The Bonferroni correction will give very conservative results and will likely yield an even sparser model. Whether one should use such adjustment is still an open question in practice.



\section{Conclusions}

 This paper proposes a \textit{p}-value for the LS-Boost algorithm in high-dimensional VARs that adapts to its iterative nature. Our \textit{p}-value can be used for incremental hypothesis testing that tests the statistical significance of every selected parameter at each boosting step. We derive the asymptotic distribution for the estimator for a given boosting step and also discuss its convergence when the boosting step is large. Simulation results indicate a \textit{p}-value-adjusted model can improve the FPR and F score of a model produced by LS-Boost. The application further reveals that the proposed \textit{p}-value can help control both model size and stability, making model interpretation relatively easy in  high dimensions.
 
 LS-Boost is a classical regression algorithm and our results complement its wide use in practice. The proposed \textit{p}-value is conceptually simple and technically practical, and our R package \texttt{boostvar} also makes it computationally easy. In addition, when both sample size and boosting step are large, we show the \textit{p}-value can converge to that of LS estimator in both VAR and cross-section regression (see tables in the online supplement), providing an implicit computational guarantee for the \textit{p}-value.
 
It will be interesting to design a procedure to properly control the FDR based on the sequential \textit{p}-values in a VAR. In addition, more applications in both cross-section regression and VAR are needed to further study the property of the \textit{p}-value. We leave these topics for future work.

\section*{Acknowledgments}

The author thanks the Department of Economics, Finance, and Quantitative Analysis for financial support and the Office of Research at Kennesaw State University for computation support.

\newpage
\spacing{1.45}
\bibliographystyle{ecca}
\bibliography{reference}

\newpage
\setcounter{page}{1}
\spacing{1.42}

\begin{appendices}
{
	
	\centering \title{\large\MakeUppercase{Supplementary Material to ``Boosted p-Values for High-Dimensional Vector Autoregression"}\footnote{Email: xhuang3@kennesaw.edu}\\[10pt]} }
	\begin{center}
		\large
		\author{\textsc{Xiao Huang}}\\
		\date{\today}
	\end{center}
	\maketitle
    
\bigskip    
    
This supplement contains all proofs, additional discussions, and tables

\setstretch{2}
\bigskip

\setstretch{1.3}
\localtableofcontents
\setstretch{2}

\newpage 
\setstretch{1.35}
\setcounter{section}{19}
\setcounter{equation}{0}
\renewcommand{\theequation}{S.\arabic{equation}}
\subsection{Proofs} \label{supp:asymptotic}
\subsubsection{Proof of \Cref{thm:group boosting asymp}}
\begin{proof}[Proof of \Cref{thm:group boosting asymp}]
	\Cref{eq:phi jk3} offers a starting point to compute the variance of $\hat{\phi}_{(j)}^{(k)\prime}$.
	\begin{equation} \label{eq:phi jk4}
		\hat{\phi}_{(j)}^{(k)\prime} = \tilde{\mathbf{A}}_{j}^{(k)} \mathbf{X}_{(1)} \phi_{(1)}'+ \cdots + \tilde{\mathbf{A}}_{j}^{(k)} \mathbf{X}_{(d)} \phi_{(d)}' + \tilde{\mathbf{A}}_{j}^{(k)} \mathbf{u}.
	\end{equation}
	Consider the term $\tilde{\mathbf{A}}_{j}^{(k)} \mathbf{X}_{(1)} \phi_{(1)}'$.
	 \begin{equation} \label{eq:phi term1}
	 	\tilde{\mathbf{A}}_{j}^{(k)} \mathbf{X}_{(1)} \phi_{(1)}' = \left[\nu \mathbf{A}_j^{(1)} + \nu \mathbf{A}_j^{(2)}(\mathbf{I}_T - \nu \mathbf{H}^{(1)})+ \cdots + \nu \mathbf{A}_j^{(k)}(\mathbf{I}_T - \nu \mathbf{H}^{(k-1)})\cdots(\mathbf{I}_T - \nu \mathbf{H}^{(1)})\right] \mathbf{X}_{(1)} \phi_{(1)}'.
	 \end{equation}
	Assume all $\mathbf{A}_j^{(k)}$ are nonzero. If $\mathbf{A}_j^{(k)} = 0$, the corresponding term will be dropped from the summation. If $\mathbf{A}_j^{(k)} \neq 0$, we can remove the superscript $(k)$ and simply write $\mathbf{A}_j^{(k)} = \mathbf{A}_j $. Consider the last summand.
	\begin{align}
		&\nu \mathbf{A}_j^{(k)}(\mathbf{I}_T - \nu \mathbf{H}^{(k-1)})\cdots(\mathbf{I}_T - \nu \mathbf{H}^{(1)}) \mathbf{X}_{(1)} \phi_{(1)}' \nonumber\\
		&= \nu \mathbf{A}_j^{(k)} \mathbf{X}_{(1)} \phi_{(1)}' + \cdots +  (-\nu)^{k} \mathbf{A}_j^{(k)} \mathbf{H}^{(k-1)} \cdots \mathbf{H}^{(1)} \mathbf{X}_{(1)} \phi_{(1)}' \nonumber\\
		&= \nu \mathbf{A}_j \mathbf{X}_{(1)} \phi_{(1)}' + \cdots +  (-\nu)^{k} \mathbf{A}_j \mathbf{H}^{(k-1)} \cdots \mathbf{H}^{(1)} \mathbf{X}_{(1)} \phi_{(1)}' \nonumber\\
		&= \nu (\mathbf{X}_{(j)}'\mathbf{X}_{(j)})^{-1}\mathbf{X}_{(j)}' \mathbf{X}_{(1)} \phi_{(1)}' + \cdots \nonumber\\
		&\quad + (-\nu)^{k} (\mathbf{X}_{(j)}'\mathbf{X}_{(j)})^{-1}\mathbf{X}_{(j)}' \cdot \mathbf{X}_{(j_{k-1})}(\mathbf{X}_{(j_{k-1})}'\mathbf{X}_{(j_{k-1})})^{-1}\mathbf{X}_{(j_{k-1})}' \cdots \mathbf{X}_{(j_1)}(\mathbf{X}_{(j_1)}'\mathbf{X}_{(j_1)})^{-1}\mathbf{X}_{(j_1)}' \mathbf{X}_{(1)} \phi_{(1)}' \nonumber\\
		&=\nu \left(\frac{\mathbf{X}_{(j)}'\mathbf{X}_{(j)}}{T}\right)^{-1} \left(\frac{\mathbf{X}_{(j)}' \mathbf{X}_{(1)}}{T}\right) \phi_{(1)}' + \cdots + (-\nu)^{k} \left(\frac{\mathbf{X}_{(j)}'\mathbf{X}_{(j)}}{T}\right)^{-1} \frac{\mathbf{X}_{(j)}'  \mathbf{X}_{(j_{k-1})}}{T} \left(\frac{\mathbf{X}_{(j_{k-1})}'\mathbf{X}_{(j_{k-1})}}{T}\right)^{-1} \nonumber\\
		&\quad \cdots \left(\frac{\mathbf{X}_{(j_1)}'\mathbf{X}_{(j_1)}}{T}\right)^{-1}\frac{\mathbf{X}_{(j_1)}' \mathbf{X}_{(1)}}{T} \phi_{(1)}' \nonumber\\
		&\rightarrow \text{some constant} \text{ as } T \rightarrow \infty. \label{eq:phi term1 proof}
	\end{align}
	The last line follows because $\mathbf{X}_{(j)}$ collects $p$ columns of lags for the $j$th variable and the vector process in $\mathbf{X}_{(j)}$ are ergodic for both the first and second moments --- a standard result in time series (see, for example, Propositions 10.2 and 10.5 in \cite{hamilton1994time}). Consequently, the r.h.s. of \cref{eq:phi term1} will be a constant as $T \rightarrow \infty$. Similarly, $\tilde{\mathbf{A}}_{j}^{(k)} \mathbf{X}_{(2)} \phi_{(2)}',\cdots,\tilde{\mathbf{A}}_{j}^{(k)} \mathbf{X}_{(d)} \phi_{(d)}'$ all converge to some constant as $T \rightarrow \infty$. Immediately, we have, for \cref{eq:phi jk4},
	\begin{equation} \label{eq:phi jk5}
		\hat{\phi}_{(j)}^{(k)\prime} = c_{(j)}^{(k)} + \tilde{\mathbf{A}}_{j}^{(k)} \mathbf{u} + o_p(1),
	\end{equation}
	where $c_{(j)}^{(k)}$ is a constant matrix with
	\begin{equation} \label{eq:phi jk5 cjk term}
		\tilde{\mathbf{A}}_{j}^{(k)} \mathbf{X}_{(1)} \phi_{(1)}'+ \cdots + \tilde{\mathbf{A}}_{j}^{(k)} \mathbf{X}_{(d)} \phi_{(d)}' \rightarrow c_{(j)}^{(k)}.
	\end{equation}
	\Cref{eq:phi jk5} greatly simplifies the computation of $\text{Var}(\hat{\phi}_{(j)}^{(k)\prime})$ and allows us to focus on the last term $\tilde{\mathbf{A}}_{j}^{(k)} \mathbf{u}$.
	
	Next, omitting the $o_p(1)$ term in \cref{eq:phi jk5}, we write
	\begin{align*}
		&\hat{\phi}_{(j)}^{(k)\prime} - c_{(j)}^{(k)} = \nu \mathbf{A}_j^{(1)} \mathbf{u}\\
		&\quad +\nu \mathbf{A}_j^{(2)} (\mathbf{I} - \nu \mathbf{H}^{(1)}) \mathbf{u}\\
		&\quad +\nu \mathbf{A}_j^{(3)} (\mathbf{I} - \nu \mathbf{H}^{(2)}) (\mathbf{I} - \nu \mathbf{H}^{(1)}) \mathbf{u}\\
		&\quad +\cdots\\
		&\quad + \nu \mathbf{A}_j^{(k)} (\mathbf{I} - \nu \mathbf{H}^{(k-1)}) \cdots (\mathbf{I} - \nu \mathbf{H}^{(1)}) \mathbf{u}\\
		&= \nu (\mathbf{X}_{(j_1)}'\mathbf{X}_{(j_1)})^{-1}\mathbf{X}_{(j_1)}' \mathbf{u} \\
		&\quad + \nu (\mathbf{X}_{(j_2)}'\mathbf{X}_{(j_2)})^{-1}\mathbf{X}_{(j_2)}'(\mathbf{I} - \nu \mathbf{H}^{(j_1)}) \mathbf{u}\\
		&\quad +\nu (\mathbf{X}_{(j_3)}'\mathbf{X}_{(j_3)})^{-1}\mathbf{X}_{(j_3)}'(\mathbf{I} - \nu \mathbf{H}^{(j_2)}) (\mathbf{I} - \nu \mathbf{H}^{(j_1)}) \mathbf{u}\\
		&\quad \cdots \\
		&\quad +\nu (\mathbf{X}_{(j_{k})}'\mathbf{X}_{(j_{k})})^{-1}\mathbf{X}_{(j_{k})}'(\mathbf{I} - \nu \mathbf{H}^{(j_{k-1})}) \cdots (\mathbf{I} - \nu \mathbf{H}^{(j_1)}) \mathbf{u}\\
		&= \nu (\mathbf{X}_{(j_1)}'\mathbf{X}_{(j_1)})^{-1}\mathbf{X}_{(j_1)}' \mathbf{u}\\
		&\quad + \nu (\mathbf{X}_{(j_2)}'\mathbf{X}_{(j_2)})^{-1}\mathbf{X}_{(j_2)}'(\mathbf{I} - \nu \mathbf{X}_{(j_1)} (\mathbf{X}_{(j_1)}'\mathbf{X}_{(j_1)})^{-1}\mathbf{X}_{(j_1)}') \mathbf{u}\\
		&\quad +\nu (\mathbf{X}_{(j_3)}'\mathbf{X}_{(j_3)})^{-1}\mathbf{X}_{(j_3)}'(\mathbf{I} - \nu \mathbf{X}_{(j_2)} (\mathbf{X}_{(j_2)}'\mathbf{X}_{(j_2)})^{-1}\mathbf{X}_{(j_2)}') (\mathbf{I} - \nu \mathbf{X}_{(j_1)} (\mathbf{X}_{(j_1)}'\mathbf{X}_{(j_1)})^{-1}\mathbf{X}_{(j_1)}') \mathbf{u}\\
		&\quad \cdots\\
	&\quad +\nu (\mathbf{X}_{(j_k)}'\mathbf{X}_{(j_k)})^{-1}\mathbf{X}_{(j_k)}'(\mathbf{I} - \nu \mathbf{X}_{(j_{k-1})} (\mathbf{X}_{(j_{k-1})}'\mathbf{X}_{(j_{k-1})})^{-1}\mathbf{X}_{(j_{k-1})}')\cdots (\mathbf{I} - \nu \mathbf{X}_{(j_1)} (\mathbf{X}_{(j_1)}'\mathbf{X}_{(j_1)})^{-1}\mathbf{X}_{(j_1)}') \mathbf{u}
	\end{align*}
	
	Again, as $T \rightarrow \infty$, most of the cross product matrices converge to a constant when divided by $T$. As a result, we can write
	\begin{align} 
		\hat{\phi}_{(j)}^{(k)\prime} - c_{(j)}^{(k)} &=c_{jk,j_{1}1} \frac{\mathbf{X}_{(j_{1})}' \mathbf{u}}{T} \nonumber \\
		&+c_{jk,j_{2}1}\frac{\mathbf{X}_{(j_{2})}' \mathbf{u}}{T} + c_{jk,j_{1}2}\frac{\mathbf{X}_{(j_{1})}' \mathbf{u}}{T} \nonumber \\
		&+c_{jk,j_{3}1}\frac{\mathbf{X}_{(j_{3})}' \mathbf{u}}{T} + c_{jk,j_{2}2}\frac{\mathbf{X}_{(j_{2})}' \mathbf{u}}{T} + c_{jk,j_{1}3}\frac{\mathbf{X}_{(j_{1})}' \mathbf{u}}{T} \nonumber \\
		&+ \cdots \nonumber \\
		&+c_{jk,j_{k}1}\frac{\mathbf{X}_{(j_{k})}' \mathbf{u}}{T} + c_{jk,j_{k-1}2}\frac{\mathbf{X}_{(j_{k-1})}' \mathbf{u}}{T} + \cdots + c_{jk,j_{1}k}\frac{\mathbf{X}_{(j_{1})}' \mathbf{u}}{T} \nonumber \\
		&= c_{jk,j_1}\frac{\mathbf{X}_{(j_{1})}' \mathbf{u}}{T} + c_{jk,j_2}\frac{\mathbf{X}_{(j_{2})}' \mathbf{u}}{T} + \cdots + c_{jk,j_k} \frac{\mathbf{X}_{(j_{k})}' \mathbf{u}}{T}, \label{eq:phi jk6}
	\end{align}
	where
	\begin{align*}
		c_{jk,j_1} &= c_{jk,j_{1}1} + \cdots + c_{jk,j_{1}(k-1)} + c_{jk,j_{1}k},\\
		c_{jk,j_2} &= c_{jk,j_{2}1} + \cdots + c_{jk,j_{2}(k-1)},\\
		\vdots\\
		c_{jk,j_k} &= c_{jk,j_{k}1}.
	\end{align*}
	Rewrite the $p \times T$ matrices $\mathbf{X}_{(j_{1})}',\cdots,\mathbf{X}_{(j_{k})}'$ in a column format so that, for example,
	\begin{equation*}
		\mathbf{X}_{(j_{1})}' = \left[\mathbf{X}_{(j_{1}),1}',\cdots,\mathbf{X}_{(j_{1}),t}',\cdots,\mathbf{X}_{(j_{1}),T}'\right],
	\end{equation*}
	and \cref{eq:phi jk6} becomes
	\begin{equation}
		\hat{\phi}_{(j)}^{(k)\prime} - c_{(j)}^{(k)} = \sum_{t=1}^{T} \left(c_{jk,j_1}\mathbf{X}_{(j_{1}),t}' + \cdots + c_{jk,j_k}\mathbf{X}_{(j_{k}),t}' \right) u_t /T,
	\end{equation}
	where $u_t$ is the $t$th row of $\mathbf{u}$ and is a $1 \times d$ row vector.
	
	Vectorization gives
	\begin{equation}
		\text{vec}(\hat{\phi}_{(j)}^{(k)\prime} - c_{(j)}^{(k)}) = \begin{bmatrix}
		\frac{1}{T}\sum_{t=1}^{T} \left(c_{jk,j_1} \mathbf{X}_{(j_{1}),t}' + \cdots + c_{jk,j_k} \mathbf{X}_{(j_{k}),t}' \right) u_{t1}\\
		\vdots \\
		\frac{1}{T}\sum_{t=1}^{T} \left(c_{jk,j_1} \mathbf{X}_{(j_{1}),t}' + \cdots + c_{jk,j_k} \mathbf{X}_{(j_{k}),t}' \right) u_{td}
		\end{bmatrix}_{dp \times 1}.
	\end{equation}
	Define
	\begin{equation*}
		\xi_t = \begin{bmatrix}
		\left(c_{jk,j_1} \mathbf{X}_{(j_{1}),t}' + \cdots + c_{jk,j_k} \mathbf{X}_{(j_{k}),t}' \right) u_{t1}\\
		\vdots \\
		\left(c_{jk,j_1} \mathbf{X}_{(j_{1}),t}' + \cdots + c_{jk,j_k} \mathbf{X}_{(j_{k}),t}' \right) u_{td}
		\end{bmatrix}
	\end{equation*}
	and we have
	\begin{equation}
			\text{vec}(\hat{\phi}_{(j)}^{(k)\prime} - c_{(j)}^{(k)}) = \frac{1}{T} \sum_{t=1}^{T} \xi_t.
	\end{equation}
	Next, we compute the first two moments of $\xi_t$. First we show $E(\xi_t) = \mathbf{0}$. To see that, it will be helpful to show, for example, $E(\mathbf{X}_{(j_{1}),t}' u_{t1}) = 0$. Given the explicit expression for $\mathbf{X}_{(j_{1})}'$,
	\begin{equation*}
		\mathbf{X}_{(j_{1})}' = \begin{bmatrix}
		y_{j_1,1-1} &\cdots &y_{j_1,t-1} &\cdots &y_{j_1,T-1}\\
		y_{j_1,1-2} &\cdots &y_{j_1,t-2} &\cdots &y_{j-1,T-2}\\
		\vdots      &\vdots &\vdots      &\vdots &\vdots\\
		y_{j_1,1-p} &\cdots &y_{j_1,t-p} &\cdots &y_{j-1,T-p}
		\end{bmatrix}_{p \times T},
	\end{equation*}
	the time subscripts of elements in its $t$th column are all less than $t$, which proves that $E(\mathbf{X}_{(j_{1}),t}' u_{t1}) = 0$. As a result, we have $E(\xi_t) = \mathbf{0}$. On the other hand, let 
	\begin{equation*}
		a_{(j),t}^{(k)} = c_{jk,j_1} \mathbf{X}_{(j_{1}),t}' + \cdots + c_{jk,j_k} \mathbf{X}_{(j_{k}),t}', 
	\end{equation*}
	and the second moment is
	\begin{align*}
		E\left(\xi_t \xi_t'\right) &= \begin{bmatrix}
		E\left(a_{(j),t}^{(k)}a_{(j),t}^{(k)\prime} \right)E\left(u_{t_1}^2\right) &\cdots &E\left(a_{(j),t}^{(k)}a_{(j),t}^{(k)\prime} \right)E\left(u_{t_1} u_{t_d}\right)\\
		\vdots &\vdots &\vdots\\
		E\left(a_{(j),t}^{(k)}a_{(j),t}^{(k)\prime} \right)E\left(u_{t_1} u_{t_d}\right) &\cdots &E\left(a_{(j),t}^{(k)}a_{(j),t}^{(k)\prime} \right)E\left(u_{t_d}^2\right)
		\end{bmatrix}\\
		&= \Omega^{(k)} \otimes \mathbf{Q}_{(j)}^{(k)},
	\end{align*}
	where
	\begin{equation} \label{eq:Omega and Q}
		\Omega^{(k)} = E\left(u_t^{(k)\prime} u_t^{(k)}\right), \mathbf{Q}_{(j)}^{(k)} = E\left(a_{(j),t}^{(k)}a_{(j),t}^{(k)\prime}\right) = \sum_{q_1 = 1}^{k} \sum_{q_2 = 1}^{k} c_{jk,j_{q_1}}E\left(\mathbf{X}_{(j_{q_1}),t}'\mathbf{X}_{(j_{q_2}),t}\right)c_{jk,j_{q_2}}',
	\end{equation}
	and $u_t^{(k)}$ is the residual at boosting step $k$.
	
    By the Law of Large Numbers for martingale difference sequence, we obtain
    \begin{equation}
	    	\frac{1}{T} \sum_{t=1}^{T} \xi_t \xi_t' \rightarrow \Omega^{(k)} \otimes \mathbf{Q}_{(j)}^{(k)}.
    \end{equation}
	The Central Limit Theory for martingale difference sequence (see Proposition 7.9 in \cite{hamilton1994time}) directly gives
	\begin{equation}
		\frac{1}{\sqrt{T}} \sum_{t=1}^{T} \xi_t \rightarrow N \left(\mathbf{0}, \Omega^{(k)} \otimes \mathbf{Q}_{(j)}^{(k)} \right).
	\end{equation}

	
\end{proof}

\subsubsection{Proof of \Cref{cor:boost_asymp}}
\begin{proof}[Proof of \Cref{cor:boost_asymp}]
	The proof follows the same steps in the proof of \Cref{thm:group boosting asymp}. By adapting the derivation for \cref{eq:phi jk5}, we have
	\begin{equation} \label{eq:phi jk5 js}
	\hat{\phi}_{(j)s}^{(k)\prime} = c_{(j)s}^{(k)} + \tilde{\mathbf{A}}_{js}^{(k)} \mathbf{u} + o_p(1),
	\end{equation}
	where $c_{(j)s}^{(k)}$ is a $1 \times d$ constant vector with
		\begin{equation}
	\tilde{\mathbf{A}}_{js}^{(k)} \mathbf{X}_{(1)s} \phi_{(1)s}'+ \cdots + \tilde{\mathbf{A}}_{js}^{(k)} \mathbf{X}_{(d)s} \phi_{(d)}' \rightarrow c_{(j)s}^{(k)} \text{ as } T \rightarrow \infty.
	\end{equation}
	
	Next, omitting the $o_p(1)$ term in \cref{eq:phi jk5 js}, we write
	\begin{align*}
	&\hat{\phi}_{(j)s}^{(k)\prime} - c_{(j)s}^{(k)} = \nu \mathbf{A}_{js}^{(1)} \mathbf{u}\\
	&\quad +\nu \mathbf{A}_j^{(2)} (\mathbf{I} - \nu \mathbf{H}^{(1)}) \mathbf{u}\\
	&\quad +\nu \mathbf{A}_j^{(3)} (\mathbf{I} - \nu \mathbf{H}^{(2)}) (\mathbf{I} - \nu \mathbf{H}^{(1)}) \mathbf{u}\\
	&\quad +\cdots\\
	&\quad + \nu \mathbf{A}_j^{(k)} (\mathbf{I} - \nu \mathbf{H}^{(k-1)}) \cdots (\mathbf{I} - \nu \mathbf{H}^{(1)}) \mathbf{u}\\
	&= \nu (\mathbf{X}_{(j_1)s_1}'\mathbf{X}_{(j_1)s_1})^{-1}\mathbf{X}_{(j_1)s_1}' \mathbf{u} \\
	&\quad + \nu (\mathbf{X}_{(j_2)s_2}'\mathbf{X}_{(j_2)s_2})^{-1}\mathbf{X}_{(j_2)s_2}'(\mathbf{I} - \nu \mathbf{H}^{(j_1)}) \mathbf{u}\\
	&\quad +\nu (\mathbf{X}_{(j_3)s_3}'\mathbf{X}_{(j_3)s_3})^{-1}\mathbf{X}_{(j_3)s_3}'(\mathbf{I} - \nu \mathbf{H}^{(j_2)}) (\mathbf{I} - \nu \mathbf{H}^{(j_1)}) \mathbf{u}\\
	&\quad +\cdots \\
	&\quad +\nu (\mathbf{X}_{(j_{k}s_k)}'\mathbf{X}_{(j_{k}s_k)})^{-1}\mathbf{X}_{(j_{k}s_k)}'(\mathbf{I} - \nu \mathbf{H}^{(j_{k-1})}) \cdots (\mathbf{I} - \nu \mathbf{H}^{(j_1)}) \mathbf{u}\\
	&= \nu (\mathbf{X}_{(j_1)s_1}'\mathbf{X}_{(j_1)s_1})^{-1}\mathbf{X}_{(j_1)s_1}' \mathbf{u}\\
	&\quad + \nu (\mathbf{X}_{(j_2)s_2}'\mathbf{X}_{(j_2)s_2})^{-1}\mathbf{X}_{(j_2)s_2}'(\mathbf{I} - \nu \mathbf{X}_{(j_1)s_1} (\mathbf{X}_{(j_1)s_1}'\mathbf{X}_{(j_1)s_1})^{-1}\mathbf{X}_{(j_1)s_1}') \mathbf{u}\\
	&\quad +\nu (\mathbf{X}_{(j_3)s_3}'\mathbf{X}_{(j_3)s_3})^{-1}\mathbf{X}_{(j_3)s_3}'(\mathbf{I} - \nu \mathbf{X}_{(j_2)s_2} (\mathbf{X}_{(j_2)s_2}'\mathbf{X}_{(j_2)s_2})^{-1}\mathbf{X}_{(j_2)s_2}')\\
	 &\qquad \cdot (\mathbf{I} - \nu \mathbf{X}_{(j_1)s_1} (\mathbf{X}_{(j_1)s_1}'\mathbf{X}_{(j_1)s_1})^{-1}\mathbf{X}_{(j_1)s_1}') \mathbf{u}\\
	&\quad +\cdots\\
	&\quad +\nu (\mathbf{X}_{(j_k)s_k}'\mathbf{X}_{(j_k)s_k})^{-1}\mathbf{X}_{(j_k)s_k}'(\mathbf{I} - \nu \mathbf{X}_{(j_{k-1})s_{k-1}} (\mathbf{X}_{(j_{k-1}s_{k-1})}'\mathbf{X}_{(j_{k-1}s_{k-1})})^{-1}\mathbf{X}_{(j_{k-1}s_{k-1})}')\cdots \\
	&\qquad \cdot (\mathbf{I} - \nu \mathbf{X}_{(j_1)s_1} (\mathbf{X}_{(j_1)s_1}'\mathbf{X}_{(j_1)s_1})^{-1}\mathbf{X}_{(j_1)s_1}') \mathbf{u}
	\end{align*}

	Following the derivation in \cref{eq:phi jk6}, we have
	\begin{align} 
	\hat{\phi}_{(j)s}^{(k)\prime} - c_{(j)s}^{(k)} 
	&= c_{jk,j_1s_1}\frac{\mathbf{X}_{(j_{1})s_1}' \mathbf{u}}{T} + c_{jk,j_2s_2}\frac{\mathbf{X}_{(j_{2})s_2}' \mathbf{u}}{T} + \cdots + c{jk,j_ks_k} \frac{\mathbf{X}_{(j_{k})s_k}' \mathbf{u}}{T}, \label{eq:phi jk6 js}
	\end{align}
	where
	\begin{align*}
	c_{jk,j_1s_1} &= c_{jk,j_{1}s_11} + \cdots + c_{jk,j_{1}s_1(k-1)} + c_{jk,j_{1}s_1k},\\
	c_{jk,j_2s_2} &= c_{jk,j_{2}s_21} + \cdots + c_{jk,j_{2}s_2(k-1)},\\
	\vdots\\
	c_{jk,j_ks_k} &= c_{jk,j_{k}s_k1}.
	\end{align*}
	Hence, we have
	\begin{equation}
		\hat{\phi}_{(j)s}^{(k)\prime} - c_{(j)s}^{(k)} = \sum_{t=1}^{T} \left(c_{jk,j_1s_1}\mathbf{X}_{(j_{1})s_1,t}' + \cdots + c_{jk,j_ks_k}\mathbf{X}_{(j_{k})s_k,t}' \right) u_t / T,
	\end{equation}
	where $\mathbf{X}_{(j_{k})s_k,t}'$ is the element on the $s_k$th row and $t$th column of the matrix $\mathbf{X}'$.
		Rewrite the above equation as
	\begin{equation}
	\hat{\phi}_{(j)s}^{(k)} - c_{(j)s}^{(k)\prime} = \begin{bmatrix}
	\frac{1}{T}\sum_{t=1}^{T} \left(c_{jk,j_1s_1}\mathbf{X}_{(j_{1})s_1,t}' + \cdots + c_{jk,j_ks_k}\mathbf{X}_{(j_{k})s_k,t}' \right) u_{t1}\\
	\vdots \\
	\frac{1}{T}\sum_{t=1}^{T} \left(c_{jk,j_1s_1}\mathbf{X}_{(j_{1})s_1,t}' + \cdots + c_{jk,j_ks_k}\mathbf{X}_{(j_{k})s_k,t}' \right) u_{td}
	\end{bmatrix}_{dp \times 1}.
	\end{equation}
	
	Define
	\begin{equation*}
	\xi_t = \begin{bmatrix}
	\left(c_{jk,j_1s_1}\mathbf{X}_{(j_{1})s_1,t}' + \cdots + c_{jk,j_ks_k}\mathbf{X}_{(j_{k})s_k,t}' \right) u_{t1}\\
	\vdots \\
	\left(c_{jk,j_1s_1}\mathbf{X}_{(j_{1})s_1,t}' + \cdots + c_{jk,j_ks_k}\mathbf{X}_{(j_{k})s_k,t}' \right) u_{td}
	\end{bmatrix}
	\end{equation*}
	and we have
	\begin{equation}
	\hat{\phi}_{(j)}^{(k)} - c_{(j)}^{(k)\prime} = \frac{1}{T} \sum_{t=1}^{T} \xi_t.
	\end{equation}
	Similar to the proof of \Cref{thm:group boosting asymp}, let
	\begin{equation*}
	a_{(j)s,t}^{(k)} = c_{jk,j_1s_1} \mathbf{X}_{(j_{1})s_1,t}' + \cdots + c_{jk,j_ks_k} \mathbf{X}_{(j_{k})s_k,t}'.
	\end{equation*}
	The first moment of  $\xi_t$ is zero and the second moment of $\xi_t$ is given by
	\begin{equation*}
	E\left(\xi_t \xi_t'\right) = \Omega^{(k)} \cdot \mathbf{Q}_{(j)s}^{(k)},
	\end{equation*}
	with $\mathbf{Q}_{(j)s}^{(k)} = E\left(a_{(j)s,t}^{(k)}a_{(j)s,t}^{(k)\prime}\right) = \sum_{q_1 = 1}^{k} \sum_{q_2 = 1}^{k} c_{jk,j_{q_1}s_{q-1}}E\left(\mathbf{X}_{(j_{q_1})s_{q_1},t}'\mathbf{X}_{(j_{q_2})s_{q_2},t}\right)c_{jk,j_{q_2}s_{q_2}}'.$ And $\mathbf{Q}_{(j)s}^{(k)}$ is a scalar.
	
	Again, use the Central Limit Theory for martingale difference sequence (see Proposition 7.9 in \cite{hamilton1994time}) to obtain
	\begin{equation}
	\frac{1}{\sqrt{T}} \sum_{t=1}^{T} \xi_t \rightarrow N \left(\mathbf{0}, \Omega^{(k)} \cdot \mathbf{Q}_{(j)s}^{(k)} \right).
	\end{equation}
\end{proof}

\subsubsection{Lemma 1}
Let $\tr$ be the matrix trace operator. Consider the convex quadratic optimization problem
\begin{equation} \label{eq:quadratic}
	h^* = \min_{x \in R^{pd} \times R^{d}} h(x):=\frac{1}{2}\tr(x'Qx) + \tr(q'x) + q^0,
\end{equation}
where $Q$ is a $pd \times pd$ symmetric, positive semi-definite matrix, $q$ is a $pd \times d$ matrix, and we can set the constant $q^0=0$ without loss of generality. Let $\lambda_{\text{pmin}}(Q)$ be the smallest non-zero eigenvalue of $Q$.
\begin{lemma} \label{lemma:1}
	If $h^* > -\infty$, then for any $pd \times d$ matrix $x$, there exists an optimal solution of the convex quadratic optimization problem so that
	\begin{equation} \label{eq:QP x result}
		\lVert x - x^* \rVert_2 \leq \sqrt{\frac{2(h(x) - h^*)}{\lambda_{\text{pmin}}(Q)}}
	\end{equation}
	and 
	\begin{equation} \label{eq:QP gradient result}
		\lVert \nabla h(x) \rVert \geq \sqrt{\frac{\lambda_{\text{pmin}}(Q) (h(x) - h^*)}{2}}.
	\end{equation}
\end{lemma}
\begin{proof}
	Consider the gradient
	\begin{align*}
		\nabla h(x) &= \frac{1}{2}(Q'+Q)x + q\\
		&= Qx+q \qquad \text{ because $Q$ is symmetric.}
	\end{align*}
	The first-order condition gives
	\begin{equation} \label{eq:lemma first-order cond}
		\nabla h(x) = Qx+q = \mathbf{0}_{pd \times d} \Rightarrow Qx = -q.
	\end{equation}
	Given the sparse eigendecomposition of $Q$, $Q = PDP'$ with $D$ being the diagonal matrix of non-zero eigenvalues of $Q$ and $P$ is orthonormal, $P'P=I$. Let $\tilde{x}$ one of the many solutions. We can verify that
	\begin{equation*}
		PP'q = -PP'Q\tilde{x} = -PP'PDP'\tilde{x} = -Q\tilde{x}=q.
	\end{equation*}
	Define $\hat{x}=-PD^{-1}P'q$, and it is the optimal solution since $Q\hat{x} = -q$. Substitute $\hat{x}$ into $h(x)$ to obtain the minimum value $h^* = -\frac{1}{2}q'PD^{-1}P'q$.
	Let $x^* = (I - PP')x - PD^{-1}P'q$. Similar to the proof of Proposition A.2.1 in \cite{freundetal2017boosting}, we have 
	\begin{align*}
		\lVert x - x^* \rVert_{2} &= \tr \left[(x'PP'+q'PD^{-1}P')(PP'x + PD^{-1}P'q)\right]\\
		&=\tr \left[(x'PD^{\frac{1}{2}}+q'PD^{-\frac{1}{2}})D^{-\frac{1}{2}}P'PD^{-\frac{1}{2}}(D^{\frac{1}{2}}P'x + D^{-\frac{1}{2}}P'q)\right]\\
		&\leq \frac{1}{\lambda_{\text{pmin}}(Q)}\tr \left[(x'PD^{\frac{1}{2}}+q'PD^{-\frac{1}{2}})(D^{\frac{1}{2}}P'x + D^{-\frac{1}{2}}P'q)\right]\\
		&=\frac{2}{\lambda_{\text{pmin}}(Q)}(h(x) - h(x^*)). \qquad \text{because $h^* = -\frac{1}{2}q'PD^{-1}P'q$}
	\end{align*}
	To prove \cref{eq:QP gradient result}, we use the fact that the trace operator is convex and the usual convex inequality holds.
	\begin{align*}
		h(x^*) &\geq h(x) + \tr(\nabla h(x)' (x^* - x))\\
		&\geq h(x) - \sqrt{\tr(\nabla h(x)'\nabla h(x))} \sqrt{\tr((x^* - x)'(x^* - x))}\\
		&=h(x) - \lVert \nabla h(x) \rVert_{2} \lVert x^* - x\rVert_{2}\\
		&\leq h(x) - \lVert \nabla h(x) \rVert_{2} \sqrt{\frac{2(h(x) - h(x^*))}{\lambda_{\text{pmin}}(Q)}}.
	\end{align*}
	Rearrange the above inequality gives \cref{eq:QP gradient result}.
\end{proof}

\subsubsection{Proof of \Cref{thm:computational bounds}}
Our proof draws idea from \cite{freundetal2017boosting} and we make no claims of originality for the proof. Because of time series dependence in the data and the matrix format of $\mathbf{Y}$, our proof is slightly more complicated than that in \cite{freundetal2017boosting}.

For LS-Boost1, define the following loss function at the parameter value $\bm{\phi}_g$:

\begin{align} \label{eq:boost1_loss}
	L(\bm{\phi}_g) &= \frac{1}{2Td} \lVert \mathbf{Y} - \mathbf{X}_g \bm{\phi}_g \rVert_{2}^2.
\end{align}
\begin{proof}
	Consider the variable selection at step $k$ in \cref{eq:group boost obj}. Let $L(\hat{\bm{\phi}}_g^{(k)})$ denote the loss evaluated at $\hat{\bm{\phi}}_g^{(k)}$. The first-order condition for the $j$th variable and all its lags is
	\begin{align} 
\underset{dp \times  d}{\nabla L(\hat{\bm{\phi}}_g^{(k)})} &=
\begin{bmatrix}
\frac{\partial L(\hat{\bm{\phi}}_g^{(k)})}{\partial \hat{\bm{\phi}}_{(1)}^{(k)\prime}}\\
\vdots\\
\frac{\partial L(\hat{\bm{\phi}}_g^{(k)})}{\partial \hat{\bm{\phi}}_{(j)}^{(k)\prime}}\\
\vdots\\
\frac{\partial L(\hat{\bm{\phi}}_g^{(k)})}{\partial \hat{\bm{\phi}}_{(d)}^{(k)\prime}}
\end{bmatrix}
= \begin{bmatrix}
-\frac{1}{Td} \mathbf{X}_{(1)}' (\mathbf{Y} - \mathbf{X}_g \hat{\bm{\phi}}_g)\\
\vdots\\
-\frac{1}{Td} \mathbf{X}_{(j)}' (\mathbf{Y} - \mathbf{X}_g \hat{\bm{\phi}}_g)\\
\vdots\\
-\frac{1}{Td} \mathbf{X}_{(d)}' (\mathbf{Y} - \mathbf{X}_g \hat{\bm{\phi}}_g)
\end{bmatrix}
=\begin{bmatrix}
-\frac{1}{Td} \mathbf{X}_{(1)}' \hat{\mathbf{R}}^{(k)}\\
\vdots\\
-\frac{1}{Td} \mathbf{X}_{(j)}' \hat{\mathbf{R}}^{(k)}\\
\vdots\\
-\frac{1}{Td} \mathbf{X}_{(d)}' \hat{\mathbf{R}}^{(k)}
\end{bmatrix}\nonumber\\
  &= -\frac{1}{Td} \mathbf{X}_g' \left(\mathbf{Y} - \mathbf{X}_g \hat{\bm{\phi}}_g^{(k)}\right) = -\frac{1}{Td} \mathbf{X}_g' \hat{\mathbf{R}}^{(k)}. \label{eq:gradient_for_loss}
\end{align}

Selecting a variable $j$ that gives the smallest MSE in \cref{eq:group boost obj} is equivalent to selecting a variable with a maximum of $\left\Vert \frac{\partial L(\hat{\bm{\phi}}_g^{(k)})}{\partial \hat{\bm{\phi}}_{(j)}^{(k)\prime}} \right\Vert_{2}$ since $\mathbf{X}_{(j)}$ is normalized with $\mathbf{X}_{(j)}' \mathbf{X}_{(j)} = \mathbf{I}_p$. If we define
	\vspace*{15px}
	\begin{equation} \label{eq:group_norm1}
		Td \left\Vert \nabla L(\hat{\bm{\phi}}_{(j_k)}^{(k)\prime}) \right\Vert_{2,\infty} = \max_{j \in \{1,\cdots,d\}} \left\| \mathbf{X}_{(j)}' \hat{\mathbf{R}}^{(k)}\right\|_{2},
	\end{equation}
	where the $j_{k}$th variable is selected at step $k$, we have
	\begin{equation} \label{eq:group_norm2}
		Td \left\Vert \nabla L(\hat{\bm{\phi}}_{(j_k)}^{(k)\prime}) \right\Vert_{2,\infty} =  \left\lVert \mathbf{X}_{(j_{k})}'\hat{\mathbf{R}}^{(k)} \right\rVert_{2}.
	\end{equation}

	
		From \cref{eq:group R update}, we have
	\begin{align} \label{eq:trace_BAB_term0}
	&L(\hat{\bm{\phi}}_g^{(k+1)}) = \frac{1}{2Td} \lVert \hat{\mathbf{R}}^{(k+1)} \rVert_{2}^2 \nonumber \\
	&=\frac{1}{2Td} \tr\left[\left(\hat{\mathbf{R}}^{(k)} - \nu \mathbf{X}_{(j_k)}(\mathbf{X}_{(j_k)}'\mathbf{X}_{(j_k)})^{-1}\mathbf{X}_{(j_k)}' \hat{\mathbf{R}}^{(k)}  \right)' \left(\hat{\mathbf{R}}^{(k)} - \nu \mathbf{X}_{(j_k)}(\mathbf{X}_{(j_k)}'\mathbf{X}_{(j_k)})^{-1}\mathbf{X}_{(j_k)}' \hat{\mathbf{R}}^{(k)}  \right)\right]\nonumber\\
	&=\frac{1}{2Td} \tr\Big[\hat{\mathbf{R}}^{(k)\prime}\hat{\mathbf{R}}^{(k)} - 2\nu \hat{\mathbf{R}}^{(k)\prime}\mathbf{X}_{(j_k)}(\mathbf{X}_{(j_k)}'\mathbf{X}_{(j_k)})^{-1}\mathbf{X}_{(j_k)}'\hat{\mathbf{R}}^{(k)} \nonumber \\
	&\quad+ \nu^2 \hat{\mathbf{R}}^{(k)\prime}\mathbf{X}_{(j_k)}(\mathbf{X}_{(j_k)}'\mathbf{X}_{(j_k)})^{-1}\mathbf{X}_{(j_k)}'\hat{\mathbf{R}}^{(k)} \Big]\nonumber\\
	&= L(\hat{\bm{\phi}}_g^{(k)}) - \frac{1}{2Td}\nu(2-\nu)\tr\left(\hat{\mathbf{R}}^{(k)\prime}\mathbf{X}_{(j_k)}(\mathbf{X}_{(j_k)}'\mathbf{X}_{(j_k)})^{-1}\mathbf{X}_{(j_k)}'\hat{\mathbf{R}}^{(k)}\right) 
	\end{align}
	The matrix $\mathbf{X}_{(j_k)}(\mathbf{X}_{(j_k)}'\mathbf{X}_{(j_k)})^{-1}\mathbf{X}_{(j_k)}'$ in \cref{eq:trace_BAB_term0} can be rewritten to facilitate the proof. Since  $(\mathbf{X}_{(j_k)}'\mathbf{X}_{(j_k)})^{-1}$ is a $p \times p$ positive definite matrix, we have the decomposition $(\mathbf{X}_{(j_k)}'\mathbf{X}_{(j_k)})^{-1} = P_{j_k} \Lambda_{j_k} P_{j_k}' = \tilde{P}_{j_k} \tilde{P}_{j_k}'$, where both $P_{j_k}$ and $\tilde{P}_{j_k}$ are invertible and $\Lambda_{j_k}$ is a diagonal matrix with all positive eigenvalues of $(\mathbf{X}_{(j_k)}'\mathbf{X}_{(j_k)})^{-1}$. If we define $\tilde{\mathbf{X}}_{(j_k)} = \mathbf{X}_{(j_k)} \tilde{P}_{j_k}$, the trace term in \cref{eq:trace_BAB_term0} becomes $\tr\left(\hat{\mathbf{R}}^{(k)\prime}\tilde{\mathbf{X}}_{(j_k)}\tilde{\mathbf{X}}_{(j_k)}'\hat{\mathbf{R}}^{(k)}\right)$; accordingly, minimize the loss in \cref{eq:trace_BAB_term0} is equivalent to maximizing the trace and selecting the maximum group norm in \cref{eq:group_norm1,eq:group_norm2} when $\mathbf{X}_{(j_k)}$ is replaced with $\tilde{\mathbf{X}}_{(j_k)}$.

	Since $\tilde{\mathbf{X}}_{(j_k)}'\tilde{\mathbf{X}}_{(j_k)} = \mathbf{I}_p$, we have
	\begin{equation} \label{eq:eigenval}
		\lambda_{\text{min}}(\tilde{\mathbf{X}}_{(j_k)}'\tilde{\mathbf{X}}_{(j_k)}) = \lambda_{\text{max}}(\tilde{\mathbf{X}}_{(j_k)}'\tilde{\mathbf{X}}_{(j_k)}) = 1.
	\end{equation} 
	The matrix $\tilde{\mathbf{X}}_{(j_k)}'\tilde{\mathbf{X}}_{(j_k)}$ is a submatrix of $\tilde{\mathbf{X}}_g'\tilde{\mathbf{X}}_g$ with $\tilde{\mathbf{X}}_g = \left[\tilde{\mathbf{X}}_{(1)},\cdots, \tilde{\mathbf{X}}_{(d)}\right]$. The minimum nonzero eigenvalue of $\tilde{\mathbf{X}}_g'\tilde{\mathbf{X}}_g$ is smaller than $\lambda_{\text{min}}(\tilde{\mathbf{X}}_{(j_k)}'\tilde{\mathbf{X}}_{(j_k)})$ by the interlacing eigenvalues result in Theorem 4.3.8 in \cite{hornjohnson1985matrixanalysis}. This establish the condition $\lambda_{\text{pmin}}(\tilde{\mathbf{X}}_g'\tilde{\mathbf{X}}_g) < p$ and we can use the same analysis below equation (2.6) of \cite{freundetal2017boosting} to conclude that the linear convergence rate $\gamma < 1$.
	
	 Using $\tilde{\mathbf{X}}_{(j_k)}$, the parameter becomes $\tilde{P}_{j_k}^{-1} \phi_{(j)}^{(k)}$ but we can always obtain $\phi_{(j)}^{(k)}$ by premultiply it by $\tilde{P}_{j_k}$. For this reason, to simplify the notation, we simply assume all $\mathbf{X}_j$'s are transformed in the following proof and rewrite \cref{eq:trace_BAB_term0} as
	
	\begin{equation} \label{eq:trace_BAB_term}
		L(\hat{\bm{\phi}}_g^{(k+1)}) = L(\hat{\bm{\phi}}_g^{(k)}) - \frac{1}{2Td}\nu(2-\nu)\tr\left(\hat{\mathbf{R}}^{(k)\prime}\mathbf{X}_{(j_k)}\mathbf{X}_{(j_k)}'\hat{\mathbf{R}}^{(k)}\right).
	\end{equation}

	Consider the last matrix trace term in \cref{eq:trace_BAB_term}.
	
		\begin{align} \label{eq:trace_BAB_inequality}
	\tr\left(\hat{\mathbf{R}}^{(k)\prime}\mathbf{X}_{(j_k)}\mathbf{X}_{(j_k)}'\hat{\mathbf{R}}^{(k)}\right) 
	&=  \left\Vert \nabla L(\hat{\bm{\phi}}_{(j_k)}^{(k)\prime}) \right\Vert_{2,\infty}^2 T^2 d^2 \nonumber\\
	&\geq  \frac{1}{d} \left\Vert \nabla L_T(\hat{\bm{\phi}}_g^{(k)}) \right\Vert_{2}^2 T^2 d^2\nonumber \\
	&\geq \frac{T^2 d^2}{d} \frac{\lambda_{\text{pmin}}(\frac{1}{Td}\mathbf{X}_g'\mathbf{X}_g)(L(\hat{\bm{\phi}}_g^{(k)}) - L^*)}{2},
	\end{align}
where the last inequality follow from \Cref{lemma:1} and $L^* = L(\bm{\phi}_{g,\text{LS}})$ is the minimum of the loss function, similar to $h^*$ in \cref{eq:quadratic}.
	Substitute \cref{eq:trace_BAB_inequality} into \cref{eq:trace_BAB_term} and subtracting $L^*$ from both sides yields
	\begin{align}
		L(\hat{\bm{\phi}}^{(k+1)})-L^* &\leq (L(\hat{\bm{\phi}}^{(k)}) - L^*) \left(1-\frac{\nu(2-\nu)\lambda_{\text{pmin}}(\frac{1}{Td}\mathbf{X}_g'\mathbf{X}_g)}{4Td^2} T^2d^2\right)\nonumber\\
		&= (L(\hat{\bm{\phi}}^{(k)}) - L^*) \left(1-\frac{\nu(2-\nu)\lambda_{\text{pmin}}(\mathbf{X}_g'\mathbf{X}_g)}{4d} \right)\nonumber\\
		&=(L(\hat{\bm{\phi}}^{(k)}) - L^*) \gamma,
	\end{align}
	where $\gamma$ is the linear convergence rate and is less than $1$ as discussed above.
	
	The remaining proof follows those in \cite{freundetal2017boosting}. Given $L(\hat{\bm{\phi}}_g^{(0)}) = L(\mathbf{0}) = \frac{1}{2Td}\lVert \mathbf{Y} \rVert_{2}^2$, we have
	\begin{align} 
		L(\hat{\bm{\phi}}_g^{(k)}) - L^* &\leq (L(\hat{\bm{\phi}}_g^{(0)}) - L^*) \gamma^k \nonumber \\  
		&= \left[\frac{1}{2Td}\lVert \mathbf{Y} \rVert_{2}^2  - \frac{1}{2Td}\lVert \mathbf{Y} - \mathbf{X}_g \hat{\bm{\phi}}_{g,\text{LS}}\rVert_{2}^2  \right] \gamma^k \nonumber \\
		&= \frac{1}{2Td}  \left[2\tr\left(\mathbf{Y}' \mathbf{X}_g \hat{\bm{\phi}}_{g,\text{LS}} \right) - \tr\left( \hat{\bm{\phi}}_{g,\text{LS}}' \mathbf{X}_g' \mathbf{X}_g \hat{\bm{\phi}}_{g,\text{LS}} \right) \right] \gamma^k  \nonumber\\
		&= \frac{1}{2Td}  \left[2\tr\left(\hat{\bm{\phi}}_{g,\text{LS}}' \mathbf{X}_g'  \mathbf{X}_g \hat{\bm{\phi}}_{g,\text{LS}} \right) - \tr\left( \hat{\bm{\phi}}_{g,\text{LS}}' \mathbf{X}_g' \mathbf{X}_g \hat{\bm{\phi}}_{g,\text{LS}} \right) \right] \gamma^k \nonumber \\
		&= \frac{1}{2Td} \lVert  \mathbf{X}_g \hat{\bm{\phi}}_{g,\text{LS}}  \rVert_2^2 \gamma^k, \label{eq:log_inequality}
	\end{align}
where we use the first-order condition for the LS estimator 
$ \mathbf{X}_g'\mathbf{Y} = \mathbf{X}_g'\mathbf{X}_g \hat{\bm{\phi}}_{g,\text{LS}}$.	Combing the above result with \cref{eq:QP x result} establishes the result in \cref{eq:coef_bound}
\begin{equation*} 
	\lVert \hat{\bm{\phi}}_g^{(k)} - \hat{\bm{\phi}}^{(k)}_{g,\text{LS}} \rVert_{2} =  \sqrt{\frac{2(L(\hat{\bm{\phi}}_g^{(k)}) - L^*)}{\lambda_{\text{pmin}}(\mathbf{X}_g'\mathbf{X}_g / Td)}} \leq \frac{\lVert \mathbf{X}_g \hat{\bm{\phi}}_{\text{g,LS}} \rVert_2}{\sqrt{\lambda_{\text{pmin}}(\mathbf{X}_g'\mathbf{X}_g)}} \gamma^{k/2}.
\end{equation*}
 Finally, substituting \cref{eq:log_inequality} into the result $\lVert \mathbf{X} \hat{\bm{\phi}}_g^{(k)} - \mathbf{X}\bm{\phi}_{g,\text{LS}}^{(k)}  \rVert_2 = \sqrt{2Td(L(\hat{\bm{\phi}}_g^{(k)}) - L^*)}$ gives \cref{eq:prediction_bound}.
\end{proof}

\subsection{Examples in cross-section regression}

The method to compute \textit{p}-value in a VAR can also be used to compute the \textit{p}-value in a cross-section regression. Our R package \texttt{boostvar} provides a function \texttt{boostls} for both parameter estimation and s.e. and p-value computation in cross-section regressions. We demonstrate the calculation of \textit{p}-values with three commonly used data sets: the diabetes dataset in the R package \texttt{lars}, the prostate cancer data from \url{http://statweb.stanford.edu/~tibs/ElemStatLearn/datasets/prostate.data}, and the red wine quality data from the UCI database at \url{https://archive.ics.uci.edu/ml/datasets/wine+Quality}. In all three cases, we report the estimate, s.e. and \textit{p}-value of the LS method, LS-Boost model selected by AIC, and LS-Boost model with $2,000$ iteration steps. The learning rate is $0.1$. We rank the variables based on the \textit{p}-values from the LS solution.  The diabetes data set has 62 variables and we report only $ 10 $ of them.

Rigorously speaking, the \textit{p}-value for the AIC-selected boosting model should be adjusted for additional uncertainty, as is discussed in \Cref{sec:discussions}. We simply use the AIC-selected model for illustration purposes.  The large number of iterations $ 2,000 $ is used to show the convergence of the LS-Boost s.e. and \textit{p}-value to those of the LS estimator. We provide a brief analysis for each data set in the following.

In \Cref{tab:diabetes_table}, the \textit{p}-value is very small for the first few variables, consistent with the result of LS method. However, it starts to deviate from LS result in several other variables. In particular, the coefficient for \texttt{tc.tch} becomes significant when $ k = 2,000 $. More surprisingly, we see the sign can also change. For example, $ \hat{\beta}_{\text{hdl.tch}}  = 1188.409$ in the LS but it is $ -100.177 $ in LS-Boost and significant. A quick calculation shows that the correlation between the dependent variable and \texttt{hdl.tch} is $ -0.00369 $, suggesting the LS-Boost gets the sign right, while the LS result can be more influenced by factors such as multicollinearity. This further demonstrates the benefit of using LS-Boost even when the LS solution is unique. In this data set, at $ k=2,000 $, we still do not seem much convergence of LS-Boost s.e. and \textit{p}-value to those of the LS method, probably due to the relative large number of regressors w.r.t. the sample size of $442$ and possibly other features of the data.

\begin{table}[htp] \centering
	\begin{center}
		\caption{Regression table for the diabetes data} 
		\label{tab:diabetes_table} 
		\begin{threeparttable} 
			\begin{tabular}{lrrrrrrrrr}  
				\toprule
				&\multicolumn{3}{c}{LS method} & \multicolumn{3}{c}{LS-Boost (AIC)} & \multicolumn{3}{c}{LS-Boost (2000 steps)}\\
				\cmidrule(lr){2-4} \cmidrule(lr){5-7} \cmidrule(lr){8-10}
				variable & estimate & s.e. & \textit{p}-value & estimate & s.e. & \textit{p}-value & estimate & s.e. & \textit{p}-value \\
				\midrule
				bmi   & 460.721 & 84.601 & 0.000 & 503.774 & 34.367 & 0.000 & 503.774 & 34.367 & 0.000 \\
				map   & 342.933 & 72.447 & 0.000 & 270.955 & 41.109 & 0.000 & 323.178 & 56.251 & 0.000 \\
				sex   & -267.344 & 65.270 & 0.000 & -147.604 & 42.758 & 0.001 & -228.060 & 56.064 & 0.000 \\
				age.sex & 148.678 & 73.407 & 0.044 & 122.407 & 37.244 & 0.001 & 169.604 & 54.616 & 0.002 \\
				bmi.map & 154.720 & 86.340 & 0.074 & 98.797 & 31.222 & 0.002 & 151.653 & 55.000 & 0.006 \\
				map.glu & -133.476 & 91.314 & 0.145 & $\cdot$ & $\cdot$ & $\cdot$ & -105.633 & 60.878 & 0.083 \\
				tch.2 & 773.375 & 606.967 & 0.203 & $\cdot$ & $\cdot$ & $\cdot$ & 46.156 & 33.736 & 0.171 \\
				tc.tch & -2205.917 & 1761.843 & 0.211 & -5.013 & 5.130 & 0.329 & -136.523 & 70.227 & 0.052 \\
				glu.2 & 114.149 & 94.122 & 0.226 & 80.605 & 29.441 & 0.006 & 115.340 & 47.530 & 0.015 \\
				hdl.tch & 1188.409 & 1002.242 & 0.236 & $\cdot$ & $\cdot$ & $\cdot$ & -100.177 & 55.824 & 0.073 \\	
				\bottomrule
			\end{tabular}
			\begin{tablenotes}[flushleft]
				\setlength\labelsep{0pt}
				\item[] \textit{Notes}: There are $442$ observations and $62$ variables in the diabetes data. Only $10$ variables are reported. The dependent variable is a measure of diabetes progression one year after baseline. Results for the intercept are skipped in this table.
			\end{tablenotes}
		\end{threeparttable}
	\end{center} 
\end{table}

In \Cref{tab:prostate_table}, we see a much more consistent pattern in s.e. and \textit{p}-values between the LS and the LS-Boost estimators when $ k = 2,000 $. It clearly shows the proposed \textit{p}-value can converge to that of LS method when both $ k $ and the sample size is relatively large. Also, we observe that, for variable \texttt{pgg45}, its \textit{p}-value is $ 0.058 $, and it later changes to $ 0.292 $ when $ k $ increases. It exemplifies the unique feature of our incremental hypothesis testing --- a coefficient can be significant when estimated partially and insignificant when estimated more fully. And we want to pick up the coefficient when it is significant during model selection.

\begin{table}[htp] \centering
	\begin{center}
		\caption{Regression table for the prostate cancer data} 
		\label{tab:prostate_table} 
		\begin{threeparttable} 
			\begin{tabular}{lrrrrrrrrr}  
				\toprule
				       &\multicolumn{3}{c}{LS method} & \multicolumn{3}{c}{LS-Boost (AIC)} & \multicolumn{3}{c}{LS-Boost (2000 steps)}\\
				      \cmidrule(lr){2-4} \cmidrule(lr){5-7} \cmidrule(lr){8-10}
				variable & estimate & s.e. & \textit{p}-value & estimate & s.e. & \textit{p}-value & estimate & s.e. & \textit{p}-value \\
				\midrule
				lcavol & 0.564 & 0.088 & 0.000 & 0.496 & 0.045 & 0.000 & 0.564 & 0.087 & 0.000 \\
				svi   & 0.762 & 0.241 & 0.002 & 0.551 & 0.105 & 0.000 & 0.762 & 0.240 & 0.001 \\
				lweight & 0.622 & 0.201 & 0.003 & 0.500 & 0.111 & 0.000 & 0.622 & 0.199 & 0.002 \\
				age   & -0.021 & 0.011 & 0.058 & $\cdot$ & $\cdot$ & $\cdot$ & -0.021 & 0.011 & 0.054 \\
				lbph  & 0.097 & 0.058 & 0.098 & 0.034 & 0.018 & 0.065 & 0.097 & 0.057 & 0.092 \\
				lcp   & -0.106 & 0.090 & 0.241 & $\cdot$ & $\cdot$ & $\cdot$ & -0.106 & 0.089 & 0.235 \\
				pgg45 & 0.004 & 0.004 & 0.310 & 0.001 & 0.001 & 0.058 & 0.004 & 0.004 & 0.292 \\
				gleason & 0.049 & 0.155 & 0.752 & $\cdot$ & $\cdot$ & $\cdot$ & 0.049 & 0.148 & 0.740 \\	
				\bottomrule
			\end{tabular}
			\begin{tablenotes}[flushleft]
				\setlength\labelsep{0pt}
				\item[] \textit{Notes}: There are $97$ observations and $8$ variables in the prostate cancer data. The dependent variable is the log of prostate-specific antigen (lpsa) score. Results for the intercept are skipped in this table.
			\end{tablenotes}
		\end{threeparttable}
	\end{center} 
\end{table}

In \Cref{tab:redwine_table}, again we see the convergence in s.e. and \textit{p}-value of the LS-Boost estimator to those of the LS estimator. A similar pattern also exists: a variable such as \texttt{FA} is significant when $ k $ is small but insignificant when $ k $ is large. This data set is used in \cite{lockhart2014lassopval}.

\begin{table}[htp] \centering
	\begin{center}
		\caption{Regression table for the red wine quality data} 
		\label{tab:redwine_table} 
		\begin{threeparttable} 
			\begin{tabular}{lrrrrrrrrr}  
				\toprule
				&\multicolumn{3}{c}{LS method} & \multicolumn{3}{c}{LS-Boost (AIC)} & \multicolumn{3}{c}{LS-Boost (2000 steps)}\\
				\cmidrule(lr){2-4} \cmidrule(lr){5-7} \cmidrule(lr){8-10}
				variable & estimate & s.e. & \textit{p}-value & estimate & s.e. & \textit{p}-value & estimate & s.e. & \textit{p}-value \\
				\midrule
				alcohol & 0.276 & 0.026 & 0.000 & 0.285 & 0.014 & 0.000 & 0.278 & 0.024 & 0.000 \\
				VA & -1.084 & 0.121 & 0.000 & -1.037 & 0.065 & 0.000 & -1.083 & 0.109 & 0.000 \\
				sulphates & 0.916 & 0.114 & 0.000 & 0.814 & 0.092 & 0.000 & 0.913 & 0.111 & 0.000 \\
				TSD & -0.003 & 0.001 & 0.000 & -0.003 & 0.000 & 0.000 & -0.003 & 0.001 & 0.000 \\
				chlorides & -1.874 & 0.419 & 0.000 & -1.684 & 0.326 & 0.000 & -1.876 & 0.382 & 0.000 \\
				pH    & -0.414 & 0.192 & 0.031 & -0.353 & 0.089 & 0.000 & -0.424 & 0.165 & 0.010 \\
				FSD & 0.004 & 0.002 & 0.045 & 0.002 & 0.001 & 0.040 & 0.004 & 0.002 & 0.035 \\
				citric acid & -0.183 & 0.147 & 0.215 & $\cdot$ & $\cdot$ & $\cdot$ & -0.179 & 0.132 & 0.173 \\
				RS & 0.016 & 0.015 & 0.276 & 0.001 & 0.001 & 0.531 & 0.016 & 0.014 & 0.270 \\
				FA & 0.025 & 0.026 & 0.336 & 0.003 & 0.001 & 0.001 & 0.023 & 0.021 & 0.266 \\
				density & -17.881 & 21.633 & 0.409 & $\cdot$ & $\cdot$ & $\cdot$ & -16.354 & 18.106 & 0.366 \\	
				\bottomrule
			\end{tabular}
			\begin{tablenotes}[flushleft]
				\setlength\labelsep{0pt}
				\item[] \textit{Notes}: There are $1599$ observations and $11$ variables in the red wine quality data. The dependent variable is a wine quality score between $0$ and $10$. The 2nd, 4th, 7th, 9th and 10th variables are volatile acidity (VA), total sulfur dioxide (TSD), free sulfur dioxide (FSD), residual sugar (RS), and fixed acidity (FA), respectively. Results for the intercept are skipped in this table.
			\end{tablenotes}
		\end{threeparttable}
	\end{center} 
\end{table}

\end{appendices}


\end{document}